%% file: arxiv-LSFlow.tex
\title{
  Pauli Fusion: a Computational Model \\
  to Realise Quantum Transformations from ZX Terms
}
\def\titlerunning{%
  The Pauli Fusion Model
}
\def\fudge{\hspace*{-2ex}}
\def\authorrunning{%
  N. de Beaudrap, R. Duncan, D. Horsman, and S. Perdrix
}
\author{%
  \fudge Niel de Beaudrap \fudge
  \institute{%
    \fudge
    University of Oxford
    \fudge \\
    Oxford, UK
  }%
  \email{niel.debeaudrap}
  \vspace*{-1ex}
  \email{@cs.ox.ac.uk}%
  \and
  \fudge Ross Duncan \fudge
  \institute{%
    \fudge
    University of Strathclyde
    \fudge \\
    Glasgow, UK
  }%
  \institute{%
    \fudge
    Cambridge Quantum
    \fudge \\ \fudge
    Computing Ltd.,
    \fudge \\
    Cambridge, UK
  }%
  \email{ross.duncan}
  \vspace*{-1ex}
  \email{@strath.ac.uk}%
  \and
  \fudge Dominic Horsman \fudge
  \institute{%
    \fudge
    Universit\'e Grenoble Alpes
    \fudge \\
    Grenoble, France
  }%
  \email{dom.horsman}
  \vspace*{-1ex}
  \email{@gmail.com}%
  \and
  \fudge Simon Perdrix \fudge
  \institute{%
    \fudge
    CNRS LORIA, Inria Mocqua
    \fudge \\
    Universit\'e de Lorraine \\
    Nancy, France
  }%
  \email{simon.perdrix}
  \vspace*{-1ex}
  \email{@loria.fr}%
}

\documentclass[submission,copyright,creativecommons]{eptcs}
\providecommand{\event}{QPL 2019} 

\usepackage{hyperref}

\def\arXiv[#1]{%
  \href{https://arxiv.org/abs/#1}{\textsf{[arXiv:#1]}}%
}

\usepackage{amssymb,amsmath,amsthm,stmaryrd,bbm,enumitem,mathtools,amsmath,subcaption}
\usepackage{framed,multicol}

\newtheorem{definition}{Definition}
\newtheorem{theorem}{Theorem}
\newtheorem{lemma}{Lemma}

\theoremstyle{definition}

\newcommand{\ket}[1]{|{#1}\rangle}
\newcommand{\bra}[1]{\langle{#1}|}
\newcommand{\denote}[1]{
\left\llbracket #1 \right\rrbracket}

\let\preceq\preccurlyeq

\def\tpreceq{\mathrel{\tilde{\phantom a}\mspace{-16mu}\preceq}}
\newcommand{\ssur}[1]{^{\smash{(#1)}}}

\newcommand\PFus{\mathrm{PF}}
\newcommand\PFcompilation{\mbox{\scshape\small\uppercase{PF-Compilation}}}
\newcommand\PFflowfinding{\mbox{\scshape\small\uppercase{PF-Flow finding}}}
\newcommand\ie{\emph{i.e.}}
\newcommand\eg{\emph{e.g.}}

\usepackage{tikzit}
\input{zx.tikzstyles}
\usetikzlibrary{calc}

\pgfdeclarelayer{edgelayer}
\pgfdeclarelayer{nodelayer}
\pgfsetlayers{background,edgelayer,nodelayer,main}
\tikzstyle{none}=[inner sep=0mm]
\tikzstyle{every loop}=[]

\newcommand\pffont[1]{\textsf{\upshape #1}}

\newcommand{\mergH}[1]{\pffont{HMerge}_{#1}}
\newcommand{\projH}[1]{\pffont{HProj}_{#1}}
\newcommand{\splitH}[1]{\pffont{HSplit}_{#1}}
\newcommand{\initH}[1]{\pffont{HInit}_{#1}}
\newcommand{\rotH}[2]{\pffont{HRot}^{#2}_{#1}}

\newcommand{\Had}{\pffont{Had}}

\newcommand{\mergV}[1]{\pffont{VMerge}_{#1}}
\newcommand{\projV}[1]{\pffont{VProj}_{#1}}
\newcommand{\splitV}[1]{\pffont{VSplit}_{#1}}
\newcommand{\initV}[1]{\pffont{VInit}_{#1}}
\newcommand{\rotV}[2]{\pffont{VRot}^{#2}_{#1}}

\DeclareMathOperator\Odd{Odd}

\newcommand\nodetype[1]{\texttt{\upshape #1}}

\newcommand\Corr[1][]{%
  \def\@tempa{#1}%
  \ifx\empty\@tempa
    \def\@tempa{\nodetype{Corr}}%
  \else
    \def\@tempa{\nodetype{Corr}^{\{s_{#1}\}}}%
  \fi\@tempa}

\newlist{algenum}{enumerate}{9}
\setlist[algenum,1]{%
  label=\arabic*.,
  itemsep=0ex,
  topsep=1ex,
  leftmargin=2em}
\setlist[algenum,2]{%
  label=\alph*\upshape.,
  itemsep=0ex,
  topsep=-0.375ex,
  leftmargin=1.5em}
\setlist[algenum,3]{%
  label=(\!\!\;\itshape\roman*\;\!\upshape),
  topsep=0ex,
  itemsep=0.375ex,
  leftmargin=2.5em}

\newcommand{\idop}{{\:\!\text{\large$\mathbbm 1$}\:\!\!}}

\newcommand\parStyle[1]{\textrm{\mdseries\upshape({#1}\kern0.1ex)}}
\newenvironment{romanum}{%
	\begin{enumerate}[label=\parStyle{\itshape\roman*},labelwidth=\romanumlabelwd,leftmargin=1\romanumlabelwd,itemsep=0.2ex]%
}{%
	\end{enumerate}%
}
\newlength\romanumlabelwd
\begin{document}

\date{}
\maketitle
\vspace*{-2ex}
\begin{abstract}
We present an abstract model of quantum computation, the \emph{Pauli Fusion} model, whose primitive operations correspond closely to generators of the ZX calculus (a formal graphical language for quantum computing).
The fundamental operations of Pauli Fusion are also straightforward abstractions of basic processes in some leading proposed quantum technologies.
These operations have non-deterministic heralded effects, similarly to  measurement-based quantum computation.
We describe sufficient conditions for Pauli Fusion procedures to be deterministically realisable, so that it performs a given transformation independently of its non-deterministic outcomes. This provides an operational model to realise ZX terms beyond the circuit model. 
\end{abstract}

\vspace{-.75em}
%
%
\section{Introduction}
%
%

Quantum computing technology is now a reality, albeit not yet fault tolerant~\cite{IBM,google}. These computers need software, for algorithm design, verification, and compilation. Previously, quantum protocols have been represented largely using circuit notation, but with actual quantum devices to now compare theory with, the shortcomings of this model becomes clear. The circuit  model in particular does not directly represent many basic physical operations, \eg,~in quantum optics~\cite{kok2009five}; it is inflexible, in that circuits cannot easily be re-written to equivalent ones; and they are computationally hard to verify. Especially in near-term noisy intermediate scale quantum (NISQ) devices, these properties give bloated software with only basic optimisation tools, which cannot be verified at scales of more than a few tens of qubits.

Recent work has placed a different way of representing quantum processes at the forefront of optimisation, verification, and design for NISQ devices and beyond. The ZX calculus \cite{interacting,zxbook} is a diagrammatic notation equipped with equational re-write rule sets, that are complete for Clifford \cite{backens2014zx,pivoting,minimal-stabiliser-zx}, Clifford+T \cite{jeandel2018complete}, and full \cite{HNW,JPV-universal,vilmart2018near,ZXNormalForm} pure-state qubit quantum mechanics, and extensible to non pure quantum evolutions \cite{carette2019completeness}. The ZX calculus has led to optimisation strategies that out-perform all others in gate compilation \cite{duncan2019graph,kissinger2019cnot}, and in T-count reduction \cite{BBW-2020} (an important metric for fault-tolerant computing). The generators of the calculus correspond closely to the basic operations of lattice surgery in the surface code~\cite{de2017zx}, which otherwise are awkward to describe using the circuit model; and ZX has been used to verify and find novel error correction procedures \cite{duncan2014steane,de2017zx,gidney2018efficient}. It comes with a scalable notation capable of representing repeated structures at arbitrary qubit scales \cite{chancellor2016graphical}. The calculus also acts in the crucial role of an intermediate representation in a new commercial quantum compiler \cite{cowtan2019qubit}.

With the success of the ZX calculus as a tool for design, verification, and optimisation of quantum operations, the question now remains: what \textit{computational model} does the ZX calculus works best with? A computational model, in the sense we use it here, is a set of primitive operations and their composition, with which one may write algorithms and protocols, and represents the information processing capabilities at the designer's disposal. If we are to get the best out of using the ZX calculus as an intermediate representation, then it is most efficient to use a computational model that reflects, and is reflected by, the basic structure of the calculus. The model will be used both at the top of the compiler stack, to conceptualise, design, and verify protocols, and also at the bottom to extract operational procedures from a ZX calculus diagram. Previously, this `operational extraction' was known only as `circuit extraction' (e.g. as in \cite{kissinger2019cnot}). The use of circuit-model gates at extraction is  wasteful, especially if is then turned into procedures such as lattice surgery that very closely model ZX generators. The key then, is to produce a computational model that works efficiently and conceptually clearly with the ZX calculus.

In this paper we introduce the \textit{Pauli Fusion} model of quantum computing, whose fundamental operations include the merging and splitting of logical (Pauli) operators. These fundamental elements very closely model those found in lattice surgery \cite{lattice}, and in operations of optical fusion gates that have similar effects \cite{kok2009five} (see Appendix \ref{app:pf-ls-of}).
The PF model takes these operations as its primitives, from which (if desired) elements in a circuit or MBQC model could be constructed.
Complementarity, not unitarity, is the guiding principle; and Pauli Fusion (PF) procedures are in general non-deterministic. 
We show how the PF model has a direct representation in terms of generators of annotated ZX diagrams. As the diagrams can include indeterminism, we give a definition and a procedure for finding the \textit{PF flow} of a ZX diagram, which is sufficient to describe how to deterministically realise an operator by a series of (individually nondeterministic) PF operations.
By `splitting the atom' of well known logic gates, as a composition of more fundamental  operations, we give a novel approach to extracting operational meaning from a ZX diagram, and a new computational model corresponding directly to the ZX calculus.

\vspace{-.75em}
%
%
\section{The Pauli Fusion model}\label{model}
%
%
\vspace{-.25em}

We define the elements of the Pauli fusion model of computation (primitive \textit{Pauli Fusion (PF) operations}) in terms of CPTP maps, not all of which are unitary.
These CPTP maps are described in terms of their Kraus operators, and may also be classically controlled. Intuitively, the basic operations can be viewed as the splitting and merging of Pauli logical operators for qubits, and the byproducts produced when two logical qubits merge into one (readers familiar with the procedure of lattice surgery in surface codes will find this a straightforward abstraction).
PF operations also have a representation by \textit{Pauli Fusion (PF) diagrams}, which we introduce here. PF diagrams are ``annotated ZX'' (or AZX) diagrams, which incorporate the nondeterministic effects of these operations.
The requirement that a PF diagram correspond to a set of PF operations imposes a restriction which we define as \textit{runnability}. We show that this requirement is equivalent to the PF diagram having a \textit{time ordering} of its elements. This adds another layer of structure (essentially directed edges) to PF diagrams that we will see in the next section is crucial to understanding when a ZX diagram can deterministically be implemented by a set of PF operations. 

\vspace{-.5em}
%
%
\subsection{Pauli Fusion operations}
%
%
\vspace{-.25em}

We define the following linear transformations on one- or two-qubit state vectors:
\vspace*{-1.5ex}%
\begin{multicols}{2}%
\setlength{\columnseprule}{.75pt}%
\addtocounter{equation}{1}%
\edef\eqlabelell{\theequation$\ell$}%
\addtocounter{equation}{-1}%
\label{eqn:PFKrausOptors}
\begin{subequations}%
\allowdisplaybreaks
\vspace*{-6.5ex}%
\begin{align}{}
    A_{V,0}
  \;&=\;
    \bra{\texttt+}
\\[1ex]
    A_{V,1}
  \;&=\;
    \bra{\texttt-}
\\[1ex]
    K_{V,0}
  \;&=\;
    \ket{\texttt+}\!\bra{\texttt{++}} \;+\; \ket{\texttt-}\!\bra{\texttt{--}}\label{kopp}
\\[1ex]
    K_{V,1}
  \;&=\;
    \ket{\texttt+}\!\bra{\texttt{+-}} \;+\; \ket{\texttt-}\!\bra{\texttt{-+}}\label{komp}
\\[1ex]
    R_{V}
  \;&=\;
    \exp\bigl(-\tfrac{1}{2}i\, Z\bigr)
\\
    A_{H,0}
  \;&=\;
    \bra{0}
\\[1ex]
    A_{H,1}
  \;&=\;
    \bra{1}
\\[1ex]
    K_{H,0}
  \;&=\;
    \ket{0}\!\bra{00} \;+\; \ket{1}\!\bra{11}\label{kozz}
\\[1ex]
    K_{H,1}
  \;&=\;
    \ket{0}\!\bra{01} \;+\; \ket{1}\!\bra{10}\label{kooz}
\\[1ex]
    R_{H}
  \;&=\;
    \exp\bigl(-\tfrac{1}{2}i\, X\bigr)
\end{align}%
\end{subequations}%
\end{multicols}%
\vspace*{-1.5ex}%
\noindent
We conceive of $A_{V,s}$ and $A_{H,s}$ as ``annihilators'' mapping one qubit to zero; we use these as Kraus maps of destructive measurement operations in the $X$ or $Z$ eigenbasis respectively, with the index $s$ representing the outcome.
(The maps $A_{V,s}^\dagger$ and $A_{H,s}^\dagger$ are then preparation maps of $X$ or $Z$ eigenstates.)
The maps $K_{V,s}$ and $K_{H,s}$ are maps from two-qubit states to single qubit states, projecting onto the $(-1)^s$ eigenstates of $X \otimes X$ or $Z \otimes Z$, and producing a single qubit which is an eigenstate of $X$ or $Z$.
(In the case of $K_{V,1}$ and $K_{H,1}$, this involves breaking the symmetry between the two qubits, in a way which is arbitrary but ultimately unimportant.)
The operations $R_V$ and $R_H$ are single-qubit Pauli $Z$ and $X$ rotations by one radian: exponentiating by a real-valued angle $\alpha$ in radians yields a single-qubit $R_z(\alpha)$ or $R_x(\alpha)$ rotation respectively.

We use these linear maps on state-vectors, together with the single-qubit Hadamard gate $H$ and the two-qubit SWAP gate, to define the \emph{(elementary) PF operations} as the following CPTP maps, described here as acting on states $\rho$ variously on one or two qubits:
\vspace*{-1ex}%
\begin{multicols}{2}%
\setlength{\columnseprule}{.75pt}%
\addtocounter{equation}{1}%
\edef\eqlabelell{\theequation$\ell$}%
\xdef\PFopnslabel{\theequation}%
\addtocounter{equation}{-1}%
\begin{subequations}%
\label{eqn:PFoperations}%
\allowdisplaybreaks
\vspace*{-8ex}%
\begin{align}{}
\mspace{-12mu}
  \Had {} (\rho)
  \;&=\; H\, \rho \:\! H^\dagger
\\[1ex]
\mspace{-12mu}
  \initV {} (\rho)
  \;&=\; \bigl(A_{V,0}^\dagger \!\!\;\otimes\!\!\; \idop)\!\: \rho \:\! \bigl(A_{V,0} \!\!\;\otimes\!\!\; \idop)
\\[1ex]
\mspace{-12mu}
  \projV {} (\rho)
  \;&=\;\;
  \sum_{\mathclap{s\in\{0,1\}}}\; A_{V,s}\, \rho \:\! A_{V,s}^\dagger
\\
\mspace{-12mu}
  \splitV {} (\rho)
  \;&=\;
  K_{V,0}^\dagger\, \rho \:\! K_{V,0}
\\[1ex]
\mspace{-12mu}
  \mergV {} (\rho)
  \;&=\;\;
  \sum_{\mathclap{s\in\{0,1\}}}\; K_{V,s}\, \rho  \:\!K_{V,s}^\dagger
\\
\mspace{-12mu}
  \rotV {} {\alpha,S,T} (\rho)
  \;&=\;
  R_{V}^{\Theta(\alpha,S,T)} \rho \:\! R_{V}^{-\Theta(\alpha,S,T)}
\\[1ex]
\mspace{-12mu}
  \sigma {} (\rho)
  \;&=\; (\mathrm{SWAP})\, \rho \:\! (\mathrm{SWAP})^\dagger
\\[1ex]
\mspace{-12mu}
  \initH {} (\rho)
  \,&=\, \bigl(A_{H,0}^\dagger \!\!\;\otimes\!\!\; \idop)\!\: \rho \:\! \bigl(A_{H,0} \!\!\;\otimes\!\!\; \idop)
\\[1ex]
\mspace{-12mu}
  \projH {} (\rho)
  \;&=\;\;
  \sum_{\mathclap{s\in\{0,1\}}}\; A_{H,s}\, \rho \:\! A_{H,s}^\dagger
\\
\mspace{-12mu}
  \splitH {} (\rho)
  \;&=\,
  K_{H,0}^\dagger\, \rho \:\! K_{H,0}
\\[1ex]
\mspace{-12mu}
  \mergH {} (\rho)
  \;&=\;\;
  \sum_{\mathclap{s\in\{0,1\}}}\; K_{H,s}\, \rho  \:\!K_{H,s}^\dagger
\\
\mspace{-12mu}
  \rotH {} {\alpha,S,T} (\rho)
  \,&=\,
  R_{H}^{\Theta(\alpha,S,T)} \rho \:\! R_{H}^{-\Theta(\alpha,S,T)}
  \tag{\eqlabelell}
\end{align}%
\end{subequations}%
\end{multicols}%

\vspace*{-1ex}
\noindent
Note that the maps $\projV {}$, $\projH {}$, $\mergV{}$, and $\mergH{}$ all are non-unitary.
The first two of these are in fact measurement operations, which we suppose also produce a classical bit as a side-effect, representing the measurement outcome (the bit $s$ which indexes the Kraus operator).
We also suppose that $\mergV {}$ and $\mergH {}$ produce a classical bit $s$ as a side-effect, indicating which of the two Kraus operators were realised.
The probability with which a given value $s \in \{0,1\}$ is realised is  determined by the square of the Euclidean norm of the state $K_{V\!\!,\,s} \ket{\psi}$,\, $K_{H\!,\,s} \ket{\psi}$,\, $A_{V\!\!,\,s}\ket{\psi}$,\!\: or $A_{H\!,\,s} \ket{\psi}$ which would result from application of one of the Kraus operators to an input state $\ket{\psi}$.

We present $\initV {}$ and $\initH {}$ as maps on a system, but their effect is to prepare fresh qubits in the $\ket{\texttt+}$ or $\ket{0}$ state.
The maps $\splitV {}$ and $\splitH {}$ realise unitary embeddings of one qubit into two.
The operations $\rotV {} {\alpha,S,T}$ and $\rotH {} {\alpha,S,T}$ depend on sets $S$ and $T$ of labels, which indicate classical bits $s_x$ for $x \in S$ or $x \in T$ which may affect the angle of rotation (some of which may be the outcomes of the maps above), according to the function
\vspace*{-1ex}
\begin{equation}
  \Theta(\alpha,S,T)
  \;=\;
    \Biggl[ \:\!\prod_{v \in S} \; (-1)^{s_v} \!\!\:\Biggr] \!\!\:\alpha + \sum_{w \in T} s_w \!\; \pi  \,.
\end{equation}%

The operations of Eqn.~\eqref{eqn:PFoperations} may be performed in tensor product, and composed in any way which is well-typed.
For the operations $\rotV {} {\alpha,S,T}$ and $\rotH {} {\alpha,S,T}$ which may depend on classical bits, we also require that the value of the bit is determined (as an input, through a probability distribution, or through an operation which determines its value) at the time the map is performed.

%
%
\subsection{Pauli Fusion diagrams}
%
%

We may readily observe that the Kraus operations defined in Eqns.~(\PFopnslabel) have straightforward representations in the ZX calculus,
\vspace*{-1ex}%
\begin{equation}%
\label{eqn:ZXofPFKrausOptors}
\begin{aligned}{}
\mspace{-12mu}
    A_{V,0}
  &=
    \denote{\,\input{ZX-VProj0.tikz}\big.\,}\!;
&\;
    A_{V,1}
  &=
    \denote{\,\input{ZX-VProj1.tikz}\big.\,}\!;
&\;
    K_{V,0}
  &=
    \denote{\,\input{ZX-VMerge0.tikz}\big.\,}\!;
&\;
    K_{V,1}
  &=
    \denote{\,\input{ZX-VMerge1.tikz}\big.\,}\!;
\mspace{-6mu}
\\[1ex]
\mspace{-12mu}
    A_{H,0}
  &=
    \denote{\,\input{ZX-HProj0.tikz}\big.\,}\!;
&\;
    A_{H,1}
  &=
    \denote{\,\input{ZX-HProj1.tikz}\big.\,}\!;
&\;
    K_{H,0}
  &=
    \denote{\,\input{ZX-HMerge0.tikz}\big.\,}\!;
&\;
    K_{H,1}
  &=
    \denote{\,\input{ZX-HMerge1.tikz}\big.\,}\!;
\mspace{-6mu}
\\[1ex]
\mspace{-12mu}
    R_{V}^{\,\alpha}
  &=
    \denote{\,\input{ZX-VRot.tikz}\big.\,}\!;
    &\;
    R_{H}^{\,\alpha}
  &=
    \denote{\,\input{ZX-HRot.tikz}\big.\,}\!,
\mspace{-6mu}
\end{aligned}%
\end{equation}%

\vspace*{-1ex}%

\noindent
where $\denote {\,\cdot\,} $ is the standard interpretation of ZX diagrams (which in this article are read from left to right). Note that these maps, together with their adjoints, generate the ZX calculus.
Considering ZX as a potential intermediate language for quantum compilers, this close relation between the Kraus operators of Pauli Fusion and the  generators of ZX provides a tantalising prospect, of using the PF model to directly represent ZX diagrams.

To pursue this line of investigation, we consider how we might use the ZX calculus to represent linear superoperators, whose Kraus operators can be obtained (or more precisely, denoted) by composing the diagrams of Eqn.~\eqref{eqn:ZXofPFKrausOptors}.
We may then consider when such diagrams represent an operation which can be realised by a Pauli Fusion procedure.

\begin{definition}
  \label{def:PFdiagram}
  A \emph{Pauli Fusion diagram} (or PF-diagram) is an AZX diagram, with labelled vertices $V(D)$ and directed edges $E(D)$ which can be generated from the set of generators below.
  (We label these diagrams with the names of Pauli Fusion operations, \eg,~``$\mergV u$'', to indicate the operation for which the node $u$ is intended to stand.)
  This diagram is accompanied by a set $\mathcal B$ of labels of bits, $s_u \in \{0,1\}$ for $u \in \mathcal B$, which are involved in the annotations (\eg,~the sets $S$ and $T$ in $\rotV {}{\alpha,S,T}$ and $\rotH {} {\alpha,S,T}$. 
  \begin{center}
\medskip
\noindent \begin{tabular}{|cc|cc|cc|}
\hline
&&&&&\\[-1.5ex]
  $\splitV u$ &{\input{VSplit.tikz}}
&
  $\mergV u$ &{\input{VMerge.tikz}}
&
  $\rotV u {\alpha,S,T}$&{\input{VRot.tikz}}
\\[-1.5ex]&&&&&\\\hline&&&&&\\[-1.5ex]
  $\splitH u$&{\input{HSplit.tikz}}
&
  $\mergH u$&{\input{HMerge.tikz}}
& 
  $\rotH u {\alpha,S,T}$&{\input{HRot.tikz}}
\\[-1.5ex]&&&&&\\\hline&&&&&\\[-1.5ex]
  $\initV u$ &{\input{VInit.tikz}}
&
  $\initH u$ &{\input{HInit.tikz}}
&
  $\Had_u$&{\input{H.tikz}}
\\[-.1ex]&&&&&\\\hline&&&&&\\[-.5ex]
  $\projV u$ &{\input{VProj.tikz}}
&
  $\projH u$ &{\input{HProj.tikz}}
&
  $\sigma$&{\input{Swap.tikz}}
\\[-1ex]&&&&&\\\hline
\end{tabular}
\medskip
\end{center}  
\end{definition}

\paragraph{Remark.}
In the case of $\rotV {} {\alpha,S,T}$ and $\rotH {} {\alpha,S,T}$, we may write an explicit formula for the angle in place of $\Theta(\alpha,S,T)$ when that formula is simple enough: for instance, we may substitute the expression ``$\Theta(0,\emptyset,\{u\})$'' with ``$s_u \pi$''.

Following~\cite{DP-2010}, we define the semantics of the Pauli-Fusion diagrams relying on the semantics of the ZX-calculus.

\begin{definition}[Denotational semantics of Pauli Fusion diagrams]
Given a PF-diagram $D$ with a set $\mathcal B$ of index-labels for classical bits $s \in \{0,1\}^{\mathcal B}$\,:
\begin{itemize}
\item
  For a given $x \in \{0,1\}^{\mathcal B}$\,, $D(x)$ denotes the ZX-diagram where $s_b \gets x_b$ for each $b \in \mathcal B$\,.
  For a given $z \in \{0,1\}^{\mathcal B \setminus V(D)}$, let $D|_z$ be the Pauli Fusion diagram obtained by the partial assignment $s_b \gets z_b$ for all $b \in \mathcal B \setminus V(D)$.
\item
  If $\mathcal B \subseteq V(D)$, $\denote D^\natural$ denotes the superoperator $\rho \mapsto \sum_{s\in \{0,1\}^{\mathcal B}} \denote {D(s)} \rho  \denote {D(s)}^\dagger$, where $\denote {\,\cdot\,} $ is the standard interpretation of the ZX-diagrams.
\item
  For a probability distribution $p$ on $\{0,1\}^{\mathcal B \setminus V(D)}$, let $\denote D_p^\natural$ denote the superoperator
  \vspace*{-1ex}
  \begin{equation}
    \rho \;\;\mapsto \
      \sum_{\mathclap{\substack{\qquad\!\!
        s\in \{0,1\}^{\mathcal B \cap V(D)}
        \\\qquad\!\!
        z \in \{0,1\}^{\mathcal B \setminus V(D)}
      }}}
      \quad
      p(z) \, \denote {D|_z(s)} \rho  \denote {D|_z(s)}^\dagger  \;.
  \end{equation}
  \vspace*{-3ex}
  
  \noindent
  As a special case, for a fixed string $r \in \{0,1\}^{\mathcal B \setminus V(D)}$, let $\denote D_r^\natural$ denote $\denote D_p^\natural$ for $p$ the point-mass distribution on $r$.
\end{itemize}
\end{definition}%
\noindent
We define $\denote{ D }^\natural_p$ and $\denote{D}^\natural_r$ above in the case $\mathcal B \not\subseteq V(D)$ 
as we anticipate that this will be useful to describe procedures which are subject to noise or classical control.
In much of what follows below, we suppose that $\mathcal B \subseteq V(D)$: when $\mathcal B$ is not taken to be a subset of $V(D)$ we shall clearly indicate that this is the case.

\begin{definition}
  For a given PF diagram $D$, a string $x \in \{0,1\}^{\mathcal B \,\cap\, V(D)}$ is called a \emph{branch} \textup(or \emph{branch string}\textup) of the PF diagram.
  If $\mathcal B \subseteq V(D)$, we call $\denote{ D(x) }$ for such a string $x$ a \emph{branch map} \textup(and in particular, \emph{the branch map for $x$}\textup) of $D$; if $\mathcal B \not\subseteq V(D)$, then for $r \in \{0,1\}^{\mathcal B \setminus V(D)}$, we call $\denote{D|_r(x)}$ a \emph{branch map of $D$ given $r$} \textup(and in particular, \emph{the branch map of $D$ for $x$ given $r$}\textup).
\end{definition}

It will be useful to analyse PF procedures entirely in terms of PF diagrams.
However, because of the simple way in which we have defined them, not all Pauli Fusion diagrams correspond to an actual Pauli Fusion procedure.
For instance, consider the diagram
\vspace*{-0.5ex}
\begin{equation}
      D \;=\;\; \input{CNONSENSE.tikz}  \;.
\end{equation}
\vspace*{-1ex}

\noindent
This is a well-formed PF diagram, and denotes a superoperator $\mathbf D = \denote{D}^\natural$.
$\mathbf D$ is in fact a unitary CPTP map, with two equivalent Kraus operators indexed by the single bit $s_w \in \{0,1\}$.
However, the elements from which the PF diagram $D$ is composed cannot be mapped directly to a PF procedure, as the bit $s_w$ which is generated by the $\mergV w$ operation is used at the $\rotH u {s_w \pi}$ operation acting on one of its inputs.
We wish to consider under what conditions a Pauli Fusion diagram corresponds, part by part, to a Pauli Fusion procedure.

\begin{definition}
  Let $D$ be a Pauli-Fusion diagram with Kraus operators governed by bits $s_u$ for $u \in \mathcal B$.
  A \emph{time-ordering} of $D$ is a function $t: V(D) \to \mathbb N$ such that, for all $u,v \in V(D)$, 
  \begin{romanum}
  \item
    if there is a directed edge $u \to v$ in $D$, then $t(u) < t(v)$;
  \item
    if the operation at $v$ is either $\rotV v {\alpha,S,T}\!\!$ or $\rotH v {\alpha,S,T}\!$ operation for some $S, T \subseteq \mathcal B$, and $u \in S \cup T$, then $t(u) < t(v)$.
  \end{romanum}
    If $D$ has a time-ordering $t$, we say that $D$ is \emph{runnable}.
\end{definition}
\begin{lemma}
 $D$ is a runnable Pauli Fusion diagram if and only if it is the diagram of a Pauli Fusion procedure $\mathfrak P$.
\end{lemma}
\begin{proof}[Proof (sketch).]
  It is easy to show that the diagram of any Pauli Fusion procedure has a time-ordering in the above sense.
  Conversely, if $D$ has a time-ordering, then for each $\tau \in \mathbb N$, recursively form the tensor product $\mathfrak P_\tau$ of all operations associated with nodes $v \in V(D)$ with $t(v) = \tau$, together with the identity operation on any qubits which are input wires of the diagram which have not yet been acted on, or qubits which have been produced by one operation but not acted on by another.
  Let $\mathfrak P = \mathfrak P_T \circ \cdots \circ \mathfrak P_1 \circ \mathfrak P_0$ for $T = \max_{v \in V} t(v)$: then $D$ is the diagram of $\mathfrak P$.
\end{proof}

\vspace{-.75em}
%
%
\section{PF-diagram extraction}
%
%

Considering ZX diagrams as an intermediate language, we wish to consider when such a diagram $D$ can be operationally realised by a Pauli Fusion procedure --- specifically, one which realises $D$ \emph{deterministically}, in the sense that all of the Kraus operators of the procedure are proportional to $\denote{D}$.
To this end, we define a ``flow'' condition --- analogous to the flow conditions of measurement-based quantum computation~\cite{DK-2006,DKMP-2007,DKPP09,perdrix:hal-01377339} --- which suffices for such a Pauli Fusion procedure to exist.

\medskip
\noindent
We use the following graph theoretic definitions:

\medskip
\begin{definition}
In a graph $G$ (possibly with self-loops) and a vertex-set $C \subseteq V(G)$, we write $\Odd(C) \subseteq V(G)$ for the set of vertices adjacent to an odd number of elements of $C$ (where a vertex with a loop is counted as a neighbour to itself).
\end{definition}

\begin{definition}
In a graph $G$ and a partial order $\preceq$ on $V(G)$, let $\prec$ stand for the irreflexive relation $(a \preceq b) \mathbin\& (a \ne b)$.
For a vertex $u \in V(G)$, we then define the \emph{future neighbourhood $N^+(u) \subset V(G)$ of $u$}, and the \emph{past neighbourhood $N^-(u) \subset V(G)$ of $u$}, by
\begin{equation*}
\begin{aligned}
  N^+(u) := \bigl\{v\in N(u) \;\big|\; u\prec v\bigr\},
\qquad
  N^-(u) :=\bigl\{v\in N(u) \;\big|\; v\prec u\bigr\}.
\end{aligned}  
\end{equation*}
We further define the shorthand $\delta^\pm(u) := \bigl\lvert N^\pm(u) \bigr\rvert$.
\end{definition}

%
%
\subsection{Signatures of ZX diagrams}
%
%

In the following, in order to maintain a close connection to PF diagrams, we suppose that ZX diagrams have distinct labels for each node, and that each open wire is explicitly indicated as either an \emph{input} or an \emph{output} wire.
We refer to such ZX diagrams as \emph{labelled} ZX diagrams.

\begin{definition}
A (labelled) ZX-diagram is in a \emph{graph-like form} \textup(or is \emph{graph-like}\textup) if it
  is $H$-free,
  has no connections between spiders of the same colour, and
  has no parallel wires or loops on any single vertex. 
\end{definition}
\noindent
By rewriting all $H$ nodes using the Euler decomposition, condensing all spiders, and removing all loops and (pairs of) parallel edges, it is easy to show that:
\begin{lemma}\label{lem:GL}
Any ZX-diagram can be transformed into a graph-like ZX-diagram.
\end{lemma}

\noindent
To a graph-like ZX diagram $D$, we associate a corresponding \emph{signature}  $(\mathcal G_D,\mathcal I,\mathcal O,\mathcal P)$,
which will enable us to reduce certain properties of $D$ to combinatorial properties of its signature.
\begin{definition}
  For a ZX diagram $D$, the \emph{signature graph} $\mathcal G_D$ is the undirected graph obtained by:
  \begin{enumerate}[itemsep=0ex, topsep=1ex]
  \item
    Adding a vertex to the open end of each input wire of $D$;
  \item
    Adding a vertex to the open end of each output wire of $D$;
  \item
    Adding a self-loop to each vertex of $D$ whose phase is an odd multiple of $\pi/2$.
  \end{enumerate}
  Then the \emph{signature} of $D$ is a tuple $(\mathcal G_D,\mathcal I,\mathcal O,\mathcal P)$ consisting of $\mathcal G_D$, together with the sets $\mathcal I, \mathcal O \subseteq V(\mathcal G_D) \setminus V(D)$ of added end-points to input\,/\,output wires, and a set $\mathcal P \subseteq V(D)$ consisting of those vertices of $D$ whose phases are an integer multiple of $\pi/2$. 
\end{definition}

\begin{figure}
\begin{equation*}
  \begin{aligned}
  \input{EX_0.tikz}  
  \end{aligned}
  \qquad
  \text{\LARGE$\leadsto$}
  \qquad
  \begin{aligned}
  \input{EX_0_OpenGraph.tikz}
  \end{aligned}
\end{equation*}
\caption{%
  \label{fig:exampleCorrectorSignature}%
  A graph-like ZX-diagram $D$, with a corresponding signature $\mathcal G_D$.
  Also indicated are the added input vertices $\mathcal I = \{i_1\}$  and output vertices $\mathcal O = \{o_1, o_2, o_3\}$; vertices in $\mathcal P$ are black (all other vertices are white).
}
\end{figure}

\medskip
\noindent 
An example of a signature graph obtained from a ZX diagram is show in 
Figure~\ref{fig:exampleCorrectorSignature}.

%
%
\subsection{Corrector sets and PF Flows}
%
%

To realise a ZX diagram as a sequence of operations, one obstacle is the fact that some simple ZX diagrams $D_0$ --- \eg,~the maps $A_{V,0}$, $A_{H,0}$, $K_{V,0}$, and $K_{H,0}$ as denoted in Eqn.~\eqref{eqn:ZXofPFKrausOptors} --- 
do not represent trace-preserving maps on their own, and must be paired with another ZX diagram $D_1$ as in the AZX diagrams of Definition~\ref{def:PFdiagram}, representing the Kraus operators of a CPTP map.

In each case, $D_1$ differs from $D_0$ by a phase operation, which raises the question of the conditions under which such a phase operation can be corrected by adapting operations which may be performed later.
For a partial order $\preceq$ representing a (somewhat flexible) time-ordering of operations, we may consider the conditions under which this is possible for a single ZX generator.

\begin{definition}[Correctors]
  \label{def:correctors}%
  Let $(\mathcal G_D,\mathcal I,\mathcal O,\mathcal P)$ be a signature of a graph-like ZX diagram $D$, and $\preceq$ a partial order on $V(\mathcal G_D)$.
  For vertices $u,v\in V(\mathcal G_D)$ and a subset $C\subseteq V(\mathcal G_D)$, we say that $C$ is a \emph{$v$-corrector of $u$} if
    $u \in \Odd(C)$, 
    and also
    \label{cond:correctorOrder}
    $v \prec w$ for all $w \in \bigl(C \setminus \mathcal P\bigr) \cup \bigl(\Odd(C) \setminus \{u\}\bigr)$.
\end{definition}
\noindent
A $v$-corrector at $u$ describes a way that a $\pi$-phase on some node $u \prec v$ in a ZX diagram $D$ can be propagated into the ``future'' through some set of nodes $C$.
If we surround all of the nodes $t \in C$ by $\pi$-phases of the opposite colour (one on each edge to a different vertex), and if we also negate the phase on $x$, this preserves the meaning of the diagram $D$.
For a graph-like ZX diagram $D$, we then propagate those $\pi$-phases to the neighbours of $t$, where they 
As $u \in \Odd(C)$, the overall phase contributed to $u$ by this process is $\pi$.
Thus, a $\pi$-phase at $u$ is equivalent to a $\pi$-phase at all nodes $w \in \Odd(C) \setminus \{u\}$, together with a change of sign at all nodes $t \in C$.

The constraint that $v \prec w$ for (some of) the nodes $w \in C \cup \bigl(\Odd(C) \setminus \{u\} \bigr)$ is motivated by the idea that the phase on $u$ is determined by an operation (represented by the vertex $v$) in the immediate past, and must be compensated for by operations which are yet to be performed.
A vertex $w$ whose phase angle is changed by a sign, or (apart from $u$) by a shift of $\pi$, is represents an operation which must be adapted to compensate for the phase on $u$, and must therefore occur in the future of $v$.
The role of $\mathcal P$ in the definition of $u$-correctors arises as follows:
\begin{itemize}
\item
  For a vertex $w \in V(D)$ whose phase is a multiple of $\pi$, changing the sign has no effect modulo $2\pi$.
  Thus, these vertices can be included in $C$, without requiring $v \prec w$ for that \emph{particular} reason (we may still require $v \prec w$ if $w \in \Odd(C)$, as we must shift then its phase by $\pi$.)
\item
  For a vertex $w$ whose phase is an odd multiple of $\pi/2$, negating the phase is equivalent to shifting the phase by $\pi$.
  If $w \in C$, then we may ignore the change of sign if the total phase from \emph{other} elements of $C$ adjacent to it is equivalent to $\pi$.
  Using the fact that $w$ is adjacent to itself, it suffices to adjust its phase (thereby requiring $v \prec w$) only if $w \in \Odd(C)$ as well.
\end{itemize}
These observations motivate the condition that $v \prec w$ only for those vertices $w$ which either belong to $C \setminus \mathcal P$, or to $\Odd(C) \setminus \{u\}$.

We now consider conditions under which every vertex is equipped with a corrector set (using the standard definitions given at the start of the section). 
\begin{definition}[PF-Flow]
For $(\mathcal G_D,\mathcal I, \mathcal O, \mathcal P)$ a signature of a graph-like ZX diagram $D$, a \emph{PF-flow} is a triple $(\preceq, f, \mathfrak C)$ consisting of a partial order $\preceq$ on $V(\mathcal G_D)$,  a function $f : V(D)
\to V(\mathcal G_D)$ and a set $\mathfrak C \subset \mathcal \wp(V(\mathcal G_D))$, such that for all $v\in V(D)$: 
\begin{romanum}
\item
  If $u$ is adjacent to $v$ in $\mathcal G_D$, then either $u \preceq v$ (and $u \notin \mathcal O$), or $v \preceq u$ (and $u \notin \mathcal I$), or both;
\item
  If $\delta^+(v)=0$, there is a set $C_{v,v} \in \mathfrak C$ which is a $v$-corrector of $v$;
\item
  For all $u \in N^-(v) \setminus \bigl\{ f(v) \bigr\}$, there is a set $C_{u,v} \in \mathfrak C$ which is a $v$-corrector of $u$.
\end{romanum}%
\end{definition}%
\noindent
That is, any pair of neighbours have some definite ordering in $\preceq$; if $v$ is a node with no neighbours in its future (which we model as a projection onto some state of one or more qubits), we require a strategy to correct a $\pi$-phase on that node; and if $v$ is a node with more than one input (which we model as a composition of merges), we must have a strategy to correct $\pi$-phases which might accumulate on its past neighbours, possibly apart from a single distinguished past neighbour.

%
%
\subsection{Compilation of ZX diagrams to Pauli Fusion diagrams}\label{sec:compilation}
%
%

Having defined PF-Flows as a strategy for correcting phases in a Pauli Fusion procedure, resulting from the different Kraus maps as we attempt to realise different ZX generators as transformations, we consider how this information can be used to deterministically realise a ZX diagram as a transformation.
This section is dedicated to the proof of the following Theorem:

\begin{theorem}
\label{thm:PFFlowSufficient}%
For any graph-like ZX-diagram $D$ with a PF-Flow,
one may construct a runnable Pauli Fusion diagram $D_\PFus$ (with a set $\mathcal B \subseteq V(D_\PFus)$ of bit-labels) which realises $D$ in every branch: that is to say, for which $\forall x\in \{0,1\}^{\mathcal B}: \denote{D_\PFus(x)} \propto \denote D$.
\end{theorem}

\begin{figure}[p]
\vspace*{-3ex}
\begin{framed}
\smallskip
\noindent\underline{\smash{\PFcompilation.}}
\smallskip

For a graph-like ZX diagram $D$ together with a PF-Flow $(\preceq, f, \mathfrak C)$, and given the signature $(\mathcal G_D,\mathcal I, \mathcal O, \mathcal P)$ for $D$, perform the following transformations on $D$:
\begin{subequations}%
\begin{algenum}[start=0]%
\item\textbf{\sffamily%
  Orientation, inputs, and outputs.
}
  Direct the edges of $D$ consistently with the partial order $\preceq$.
  At each input and output, add a trivial node with the corresponding vertex-label $i_j \in \mathcal I$ or $o_j \in \mathcal O$, with the opposite colour to the first/final node on that input/output, \eg:
  \vspace*{-1ex}
  \begin{align} 
    {\input{first_Spider.tikz}}      
  \;&\mapsto\;\,
    {\input{first_Spider_xform.tikz}}  
&;\qquad
    {\input{final_Spider.tikz}}      
  \,&\mapsto\;
    {\input{final_Spider_xform.tikz}}  
  \end{align}
  
\item\textbf{\sffamily%
  Merge-Split Decomposition.
}
  Decompose each $v \in V(D)$ of degree $> 2$ using merges, rotations, splits, as in Eqn.~\eqref{eqn:spiderSplit}.
  If $\delta^-(v) = 0$, replace the merges with a preparation; and if $\delta^+(v) = 0$, replace the splits with a projection.
  \vspace*{-1ex}
  \begin{align}{}
    \label{eqn:spiderSplit}
    \mspace{-12mu}
    \begin{aligned}
      \input{Spider_comp.tikz}      
    \end{aligned}
  \;\;\mapsto
    \begin{aligned}
      \input{Spider_decomp.tikz}  
    \end{aligned}
    \,\,
  \\[-4ex]\notag
  \end{align}%
  Define $P_{u_i,v} := \bigl\{ v_k \,\big|\; \text{the dot with a $s_{v_k} \pi$ phase is on the path from $u_i$ to $v$}\bigr\}$.

\item\textbf{\sffamily%
  Projections.
}
  Implement each projection by a rotation and a measurement, \eg:
  \vspace*{-1ex}
  \begin{align}
    \begin{aligned}
      \input{Proj.tikz} \;\;\mapsto\;\; \input{H_rot_meas.tikz}
    \end{aligned}
  \end{align}

\item\textbf{\sffamily%
  Merge-Correction.
}
  For each $v \in V(D)$, each neighbour $u\in N^-(v)\setminus \{f(v)\}$,
  and each $t \in C_{u,v} \setminus \mathcal P$ and $w \in \Odd(C_{u,v}) \setminus \{u\}$, modify the nodes $t$ and $w$ as follows: 
  \vspace*{-1ex}
  \begin{align}
      \input{Mcorr-C.tikz} \;\;\mapsto{}&\;\; \input{Mcorr-C-2.tikz}
      \\[1ex] 
      \input{Mcorr-OddC.tikz} \;\;\mapsto{}&\;\; \input{Mcorr-OddC-2.tikz}
  \end{align}

\item\textbf{\sffamily%
  Projector-Correction.
}
  For each $v \!\in\! V\!\!\;(D)$, each neighbour $u \!\in\! N^-(v) \!\setminus\! \{f(v)\}$, and each $t \in C_{u,v} \setminus \mathcal P$ and $w \in \Odd(C_{u,v}) \setminus \{u\}$, modify the nodes $t$ and $w$ as follows: 
  \vspace*{-1ex}
  \begin{align}
      \input{Mcorr-C.tikz} \;\;\mapsto{}&\;\; \input{Pcorr-C-2.tikz}
      \\[1ex] 
      \input{Mcorr-OddC.tikz} \;\;\mapsto{}&\;\; \input{Pcorr-OddC-2.tikz}
  \end{align}

\end{algenum}
\end{subequations}
\vspace*{-3ex}
\end{framed}
\vspace*{-2ex}
\caption{%
  \label{fig:PF-compilation}
  An illustrated procedure to transform a ZX-diagram $D$ with a PF-Flow into a corresponding Pauli Fusion diagram $D_\PFus$.
}
\vspace*{-1ex}
\end{figure}

\noindent

The proof involves a procedure \PFcompilation\ to construct $D_\PFus$, shown in Figure~\ref{fig:PF-compilation}.
In the following, we occasionally refer to $u \in V(D)$ as vertex \emph{labels}, as well as vertices.
This is important because the diagram $D_\PFus$ is constructed from $D$ in such a way that every node-label in $V(D)$ is also a node-label in $V(D_\PFus)$, but in some cases with significantly different relationships to other vertices.
In particular, by the construction of \PFcompilation, any label $u \in V(D)$ corresponds to a vertex in $D_\PFus$ with degree at most two.

\begin{lemma}[Runnability]%
\label{lemma:runnability}%
Let $D_\PFus$ be the Pauli Fusion diagram which \PFcompilation\ produces from graph-like ZX diagram $D$ and a PF-Flow $(\preceq, f, \mathfrak C)$. 
Then $D_\PFus$ is runnable.
\end{lemma}

\begin{lemma}[Determinism]%
\label{lemma:determinism}%
Let $\mathcal B := \bigl\{u\in V(D_\PFus) \;\big|\; \text{$u$ is a merge or a projection}\bigr\}$. 
For any $s \in \{0,1\}^{\mathcal B}$, $\denote{D_\PFus(s)} = \pm \denote{D}$. 
\label{lem:determinism}\end{lemma}

\medskip
\noindent
The proofs for these Lemmas are given in Appendix \ref{app:run} and \ref{app:det}.

\vspace{-.75em}
%
%
\section{An efficient algorithm for find PF-Flows}
%
%

The PF-Flow is a sufficient condition for compiling a ZX-diagrams  to a Pauli Fusion diagram. In this section we show there exists an efficient algorithm for deciding whether an ZX-diagram has a PF-Flow. Like for the flow condition for measurement-based quantum computing \cite{MP-2008}, the algorithm produces a PF-Flow, when it exists, and hence a strategy for correcting phases.

\begin{figure}[!b]
\vspace*{-1.5ex}
\begin{framed}
\vspace*{-0.5ex}
\noindent\underline{\smash{\PFflowfinding}.} 
\smallskip
For the signature $(\mathcal G_D,\mathcal I, \mathcal O, \mathcal P)$ of a graph-like ZX diagram $D$:

\textbf{\sffamily Initialise} 
\begin{minipage}[t]{0.75\textwidth}
  $M := \mathcal O$, the set of marked elements; \\
  $\delta M := \mathcal O$, a set of newly marked elements; \\
  ${\preceq} := \bigl\{ (u,u) \,\big|\, u \in V(\mathcal G_D) \bigr\}$, a partial order relation on $M$; \\
  $f := \emptyset$, a function on the empty subset $\emptyset \subset V(\mathcal G_D)$; \\
  $\mathfrak C := \emptyset$, an empty set of corrector-sets.
\end{minipage}
\medskip

\textbf{\sffamily Repeat} until $\delta M=\emptyset$: 
\begin{algenum}
\item
  Reset $\delta M := \emptyset$.
\item
  \label{step:find_correctable_vtces}%
  Let $R$ be the set of vertices $u \in V(D) \setminus M$, for which
  \\
  there exists a set $C_{u} \subseteq M \cup \mathcal P$ such that $\Odd(C_{u}) \setminus M = \{u\}$.  
\item
  For each $v \in V(D) \setminus M$:
  \begin{algenum}
  \item
    Let \begin{minipage}[t]{0.85\textwidth}
      $N_v$ be a set of vertices ``nearby'' to $v$ to test for correctability: \\
      If $N(v) \cap M = \emptyset$, let $N_v := N(v) \cup \{v\}$; 
      otherwise let $N_v := N(v) \setminus M$.
    \end{minipage}
  \item
    Let $F_v := N_v \setminus R$ be the set of non-correctable vertices nearby to $v$.
  \item
    If $\lvert F_v \rvert \le 1$:
    \begin{algenum}
    \item%
      \label{step:mark_an_element}
      Add $v$ to the set of newly marked elements, $\delta M := \delta M \cup \{ v \}$.
    \item
      For 
        each $u \in N_v \cap R$: 
        let $C_{u,v} := C_u$ as constructed above, \\
        and add it to the set of corrector sets, $\mathfrak C := \mathfrak C \cup \{ C_{u,v} \}$.
    \end{algenum}
  \item
    If $F_v = \{ w \}$ for some $w \in N(v)$, let $f := f \cup \{(v,w)\}$. 
  \item
    If $F_v = \emptyset$, pick some $m \in M$ and let $f := f \cup \{(v,m)\}$.
  \end{algenum}
\item
  Update the partial order ${\preceq} := {\preceq} \,\cup\, (\delta M \times M)$,
  \\
  so that the old marked elements bound the newly marked elements from above.
\item
  \label{step:mark_vertices}
  Update the set of marked elements $M := M \cup \delta M$.
\end{algenum}
\textbf{\sffamily Return} $(\preceq, f, \mathfrak C)$ if $V(D) \subseteq M$; otherwise return $(\emptyset, \emptyset, \emptyset)$.
\vspace*{-1ex}
\end{framed}

\vspace*{-1ex}
\caption{
  A procedure to efficiently construct a PF-Flow, provided one exists.
  \label{fig:PF-flow}%
}
\end{figure}

\begin{theorem}
  Given a graph-like ZX-diagram $D$, there is an efficient algorithm to decide whether it has a PF-Flow, and to construct a PF-Flow if one exists.
\end{theorem}

Figure \ref{fig:PF-flow} presents an algorithm \PFflowfinding\ to construct a PF-Flow, if one exists.
It determines the partial order $\preceq$ and the corrector-sets $C_{u,v} \in \mathfrak C$, starting from the output, and working back to earlier elements in $\preceq$ towards the preparations and input.

\begin{lemma}%
\label{lemma:flowFindingRuntime}
  \PFflowfinding\ halts in time $\mathrm{poly}(n)$ for $n := \lvert V(D) \rvert$.
\label{lem:halt}\label{lem:itworks}\end{lemma}

\begin{lemma}%
  \label{lemma:flowFindingCompleteness}
  If $D$ is a graph-like ZX diagram with a PF-Flow, then \PFflowfinding\ constructs such a PF-Flow.  
\end{lemma}

\smallskip
\noindent
The proofs for these Lemmas are given in Appendix \ref{app:halt} and  \ref{app:suc}.

%
%
\section{Conclusions}
%
%

We have introduced the Pauli Fusion model for quantum computing, enabling us to describe the splitting and merging of information represented by Pauli observables as fundamental information processing operations, as observed in lattice surgery~\cite{de2017zx} and optical fusion~\cite{kok2009five}.
We have given the annotated ZX diagrams that directly represent such operations. Thus PF operations and diagrams represent the logic of the actual physical processing operations, fulfilling the analogous role to gates in the circuit model.

The relationship between the PF model and standard ZX is an important consideration. Our development of the Pauli Fusion model was prompted by the many potential uses of the ZX calculus in quantum computing. The issue with using standard ZX operationally, as noted in the Introduction, was to translate a ZX diagram into a set of operations to run on a device. This `circuit extraction' problem has proved both difficult and costly in terms of operational overheads. By introducing an operational representation that is very closely aligned to ZX, we solved the extraction problem almost by fiat --- but with one important issue outstanding, whether the ZX diagram could be implemented deterministically with the nondeterministic operations of Pauli Fusion. The results in this paper give our solution: when a ZX diagram has a PF-Flow then it may be implemented deterministically, in a way which can be easily obtained by suitable transformations on the ZX diagram. Moreover, we give the polytime algorithm \PFflowfinding\ that finds such a PF-Flow when one exists.

The introduction of Pauli Fusion, and its position as a native operational model for ZX, allows us to envisage using ZX to work the full stack of quantum computing --- from design, through to compilation, and then operational extraction --- without passing through the conventional circuit model.
We hope that the PF model will enable full use of the power of ZX for compilation, optimisation, and verification; and new ways of understanding how physical systems process quantum information.

\vspace*{-2ex}

\subsection*{Acknowledgements}
\vspace*{-1ex}

We are grateful to Pieter Kok for suggesting that optical fusion gates have similar effects to lattice surgery procedures.
NB is supported by the EPSRC National Hub in Networked Quantum Information Technologies (NQIT.org).
DH acknowledges financial support from the `Investissements d'avenir' (ANR-15-IDEX-02) program of the French National Research Agency.
SP acknowledges support from the projects ANR-17-CE25-0009
SoftQPro, ANR-17-CE24-0035 VanQuTe, PIA-GDN/Quantex, and LUE / UOQ. 

\vspace*{-1ex}

\bibliographystyle{eptcs}
\bibliography{generic}

\newpage
\appendix

\section{Lattice surgery and optical fusion gates as Pauli Fusion}\label{app:pf-ls-of}

In this Appendix we indicate how the basic operations of lattice surgery and optical quantum computing are reflected in the abstract elementary Pauli Fusion operations of Section \ref{model}.

\subsection{Lattice surgery}

It was shown in \cite{de2017zx} that the elementary operations of lattice surgery in the surface code \cite{lattice} have the form in terms of Krauss operators that is given here in the equations \eqref{kopp}, \eqref{komp} for a rough merge (fusing the $Z$ logical operator) and  \eqref{kozz}, \eqref{kooz} for a smooth merge (fusing the $X$ logical operator). The adjoint operations are that of the rough and smooth split, respectively.

 In lattice surgery, the probabilistic biproduct (and hence the pairs of Krauss operators) comes about because two surfaces supporting logical qubits are being fused into one. The additional degree of freedom is the comparison between the logical operators that are not being fused. Figure \ref{fig:lsmerge} shows this process for rough merging.
 
\begin{figure}
	\centering
	\includegraphics[width=8cm]{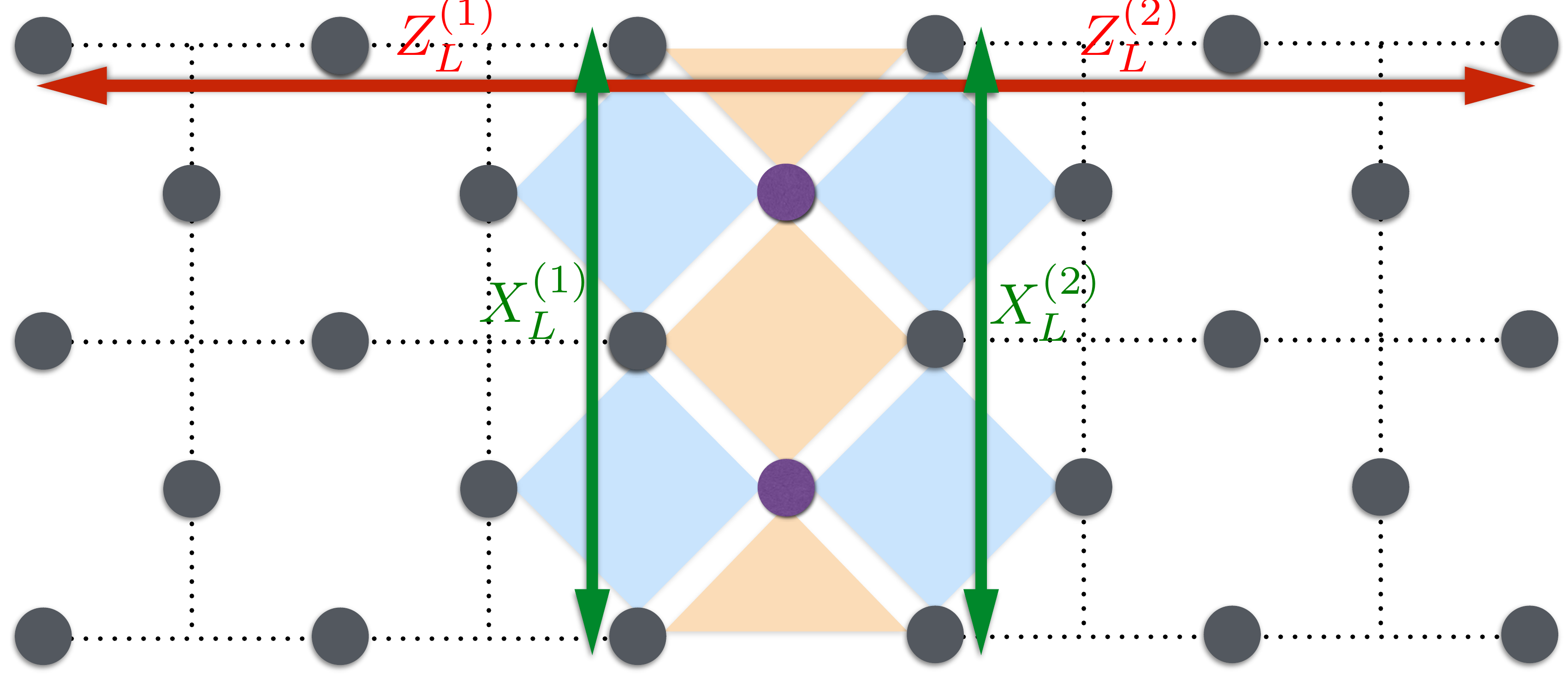}
	\caption{A rough merge. Measuring orange plaquette operators across the join fuses the $Z_L$ operators and outputs the result of a $X_L\ssur{1} \otimes X_L\ssur{2}$ measurement as a classical bit.}
	\label{fig:lsmerge}
\end{figure}

The result of this rough merge is a single qubit with fused $Z_L$ operators. If the $X_L$ operators were identical before then no correction is needed and the operation of the merge is as \eqref{kopp}: 
\begin{equation}     K =  \ket{\texttt+}\!\bra{\texttt{++}} \;+\; \ket{\texttt-}\!\bra{\texttt{--}} . \end{equation}

If they were different (heralded by the -1 measurement outcome) then correction is applied to half the surface -- equivalent to the correction being applied to one incoming surface before the merge. This results in the operation being given by \eqref{komp}: 
\begin{equation} K' = \ket{\texttt+}\!\bra{\texttt{+-}} \;+\; \ket{\texttt-}\!\bra{\texttt{-+}} \end{equation}

 In plain ZX terms these two potential merge operations are written as $\Big(\input{ZX-HMerge0.tikz};\input{ZX-HMerge1.tikz}\Big)$. These can now both be grouped as a single elementary operation of Pauli Fusion,
\begin{equation*}
\mathrm{\mergH u:} \ \ \ \input{HMerge.tikz}
\end{equation*}

\noindent where the colour-reversed diagram gives the smooth merge.

\newpage
\subsection{Optical quantum computing}

The operations of the Type I and Type II optical fusion gates \cite{browne2005resource,kok2009five} bear a close relationship to Pauli Fusion. Both fusion gates are conceived of as acting on halves of separate Bell pairs and entangling them using polarisation rotations, beam splitters, and measurement, figure \ref{fig:fusion}. They were originally developed to fuse together cluster states for one-way quantum computing.
Both are probabilistic gates -- that is, there is a set of measurement outcomes that are labelled `failure' and the system is thrown away. In the Type I gate there is only one `success' outcome, and we will see that corresponds to the positive Pauli Fusion branch (biproducts = 0) of the gate viewed as a merge. The Type 2 by contrast contains both options of a Pauli Fusion merge in the `success' branch of its operation that is not post-selected out. 

\begin{figure}
     \centering
     \begin{subfigure}{0.4\textwidth}
         \includegraphics[width=4cm]{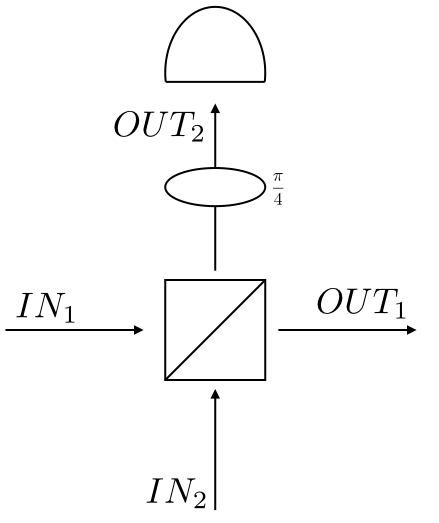}
         \caption{Type I optical fusion gate.}
         \label{fig:typeone}
     \end{subfigure}
     \hfill
     \begin{subfigure}{0.5\textwidth}
         \centering
         \includegraphics[width=5cm]{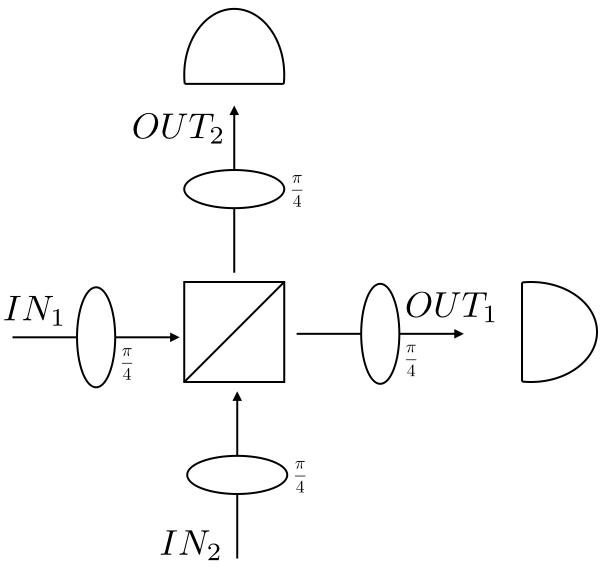}
         \caption{Type II optical fusion gate.}
         \label{fig:type2}
     \end{subfigure}
        \caption{The two types of optical fusion gates using a polarising beam splitter and detectors in the polarisation $|H,V\rangle$ basis. Both work post-selectively to entangle the inputs $IN_1$ and $IN_2$.}
        \label{fig:fusion}
\end{figure}

\subsubsection{Type I fusion gate}

The Type I gate, figure \ref{fig:typeone}, uses qubit states encoded in horizontal and vertical polarisation states of single photons, $\ket{0,1}:=\ket{H,V}$. Subscripts denote the port of the photon: e.g. $|H_1\rangle$ is a horizontally polarised photon in either the $IN_1$ or $OUT_1$ ports. The action of the polarising beam splitter is:
\begin{align}
\ket{H_1}&\rightarrow \ket{H_1}           &  \ket{V_1}&\rightarrow \ket{V_2}     \nonumber      \\
\ket{H_2}&\rightarrow \ket{H_2}         &  \ket{V_2}&\rightarrow \ket{V_1} 
\end{align}

The $\pi/4$ polarisation rotation is the Hadamard $\ket{H,V} \rightarrow \frac{1}{\sqrt{2}}( \ket{H} \pm \ket{V})$. The detector in the $OUT_2$ port is in the $\ket{H,V}$ basis.

We can therefore track what happens to the four qubit basis states as they go through the gate. Firstly, through the PBS:
\begin{align}
\ket{H_1H_2}&\ \overset{pbs}{\longrightarrow}\  \ket{H_1H_2}		 \nonumber	\\
\ket{H_1V_2}&\ \longrightarrow \ \ket{H_1V_1}		 \nonumber	\\
\ket{V_1H_2}&\ \longrightarrow \ \ket{V_2H_2}		 \nonumber	\\
\ket{V_1V_2}&\ \longrightarrow \ \ket{V_2V_1}	\label{eq:pbs}	 
\end{align}

Note that inside the gate, not everything will made immediate sense as qubit states. The photon polarisation states and number are, however, well defined.
Now we apply the $\pi/4$ rotation (neglecting normalisation for simplicity):
\begin{alignat}{2}
\ket{H_1H_2} & \overset{pbs}{\longrightarrow}  \ket{H_1H_2} &&\overset{\pi/4}{\longrightarrow}  \ket{H_1H_2}	 + \ket{H_1V_2} \nonumber \\
\ket{H_1V_2} & \longrightarrow  \ket{H_1V_1} &&\longrightarrow \ket{H_1V_1}	 			 \nonumber\\
\ket{V_1H_2} & \longrightarrow  \ket{V_2H_2} &&\longrightarrow \ket{H_2H_2}	 - \ket{V_2V_2}	 \nonumber \\
\ket{V_1V_2} & \longrightarrow  \ket{V_2V_1} &&\longrightarrow  \ket{H_2V_1}	 - \ket{V_2V_1}	 \label{eq:rot}
\end{alignat}

The post-selection criterion for the gate is that a single photon (only) must be detected in $OUT_2$. The `failure' cases are therefore the middle two rows, as the input $\ket{H_1V_2}$ produces two photon in the $OUT_1$ port and none in $OUT_2$, and the $\ket{V_1H_2}$ input produces two photons in $OUT_2$. With no loss or other error, the post-selected `success' operations (the first and last in \eqref{eq:rot}) can be represented by the operator
\begin{equation}     K = \ket{H_1}\!\bra{H_1H_2} \;+\; \ket{V_1}\!\bra{V_1V_2} \ = \  \ket{0}\!\bra{00} \;+\; \ket{1}\!\bra{11}  \end{equation}

\noindent which is our positive branch merge operator \eqref{kozz}. The post-selected action of the gate is therefore to fuse the incoming $Z$ operators of the photon qubits.

\subsubsection{Type II fusion gate}

While the Type I gate has similarities to merge operations, it is not a full Pauli Fusion operation as we have given in this paper, because there is no negative branch outcome. The Type II fusion gate includes this outcome (as well as its own `failure' outcome that is post-selected against).

The action of the Type II gate is much more complicated. The first two $\pi/4$ rotations act on the basis input states as
\begin{align}
\ket{H_1H_2}&\ \overset{\pi/4}{\longrightarrow} \ \ket{H_1H_2} + \ket{H_1V_2} +	\ket{V_1H_2}+\ket{V_1V_2}	 \nonumber	\\
\ket{H_1V_2}&\ \longrightarrow \  \ket{H_1H_2} - \ket{H_1V_2} +	\ket{V_1H_2}-\ket{V_1V_2}		 \nonumber	\\
\ket{V_1H_2}&\ \longrightarrow \  \ket{H_1H_2} + \ket{H_1V_2} -	\ket{V_1H_2}-\ket{V_1V_2}		 \nonumber	\\
\ket{V_1V_2}&\ \longrightarrow \  \ket{H_1H_2} - \ket{H_1V_2} -	\ket{V_1H_2}+\ket{V_1V_2}		 
\end{align}

For simplicity, we now track only the $\ket{00}$ input state in detail, as the others differ only by the $\pm$ signs in the superposition. The action of the PBS on each element in the superposition is the same as in \eqref{eq:rot}. The result is then
\begin{alignat}{2}
\ket{H_1H_2} & \ \overset{\pi/4}{\longrightarrow} \ \ket{H_1H_2} + \ket{H_1V_2} +	\ket{V_1H_2}+\ket{V_1V_2} &&\overset{pbs}{\longrightarrow}  \ket{H_1H_2}	 + \ket{H_1V_1} +  \ket{V_2H_2} +  \ket{V_2V_1}  
\end{alignat}

The second pair of rotations in the output ports produces (after cancellations) the final state
\begin{align}
\ket{H_1H_2} \ \xrightarrow{\pi/4 + pbs + \pi/4} \ & \ket{H_1H_2}  + \ket{V_1V_2} \nonumber\\
 & +  \ket{H_1H_1} + 2 \ket{H_1V_1} +  \ket{V_1V_1} \nonumber\\
 & +  \ket{H_2H_2} + 2 \ket{H_2V_2} +  \ket{V_2V_2}  
\end{align}

The post-selection criterion is again that one photon exactly is detected; in this case, one in each output port. The bottom two rows of this state are therefore the `failure' states, as in each case both photons exit one port. Applying the post-selection to all the input basis states gives
\begin{align}
\ket{H_1H_2}& \ \xrightarrow{\pi/4 + pbs + \pi/4 + post} \ \ket{H_1H_2}  + \ket{V_1V_2}	 \nonumber	\\
\ket{H_1V_2}&\ \xrightarrow{\hspace{6.5em}} \ \ket{H_1V_2}  + \ket{V_1H_2}	 \nonumber	\\
\ket{V_1H_2}&\ \xrightarrow{\hspace{6.5em}} \ \ket{H_1V_2}  + \ket{V_1H_2}	 \nonumber	\\
\ket{V_1V_2}&\ \xrightarrow{\hspace{6.5em}} \ \ket{H_1H_2}  + \ket{V_1V_2}
\end{align}

We see that, unlike in the Type I gate, there are no input basis states selected out. We can also see that the final measurements $HH,HV,VH,VV$ give us one piece of information: whether the inputs were the same or different. We therefore have two operators defining the evolution of the gate, depending on the measurement outcome parity (positive or negative):
\begin{equation}
 \bra{H_1H_2}  + \bra{V_1V_2} \ \ \ \mathrm{or} \ \ \ \bra{H_1V_2}  + \bra{V_1H_2}
\end{equation}

This may not look immediately Pauli Fusion-like as there is no outcome (both photons are detected and hence absorbed). However, recall that the gates act on Bell pairs. The operation in the positive-parity outcome situation is therefore
\begin{equation*}
\input{ZX-fusion02.tikz} \ = \input{ZX-fusion01.tikz} \ \ \ .
\end{equation*}
The negative branch is 
\begin{equation*}
\input{ZX-fusion12.tikz} \ = \input{ZX-fusion11.tikz}\ \ \ .
\end{equation*}

These can now both be combined as the Pauli Fusion diagram $\input{ZX-fusionPF.tikz}\ \ \ $ .\\

A further consequence is that other important operations in optical quantum computing that use such a parity projection (including entanglement swapping, entanglement distillation, and the Barrett-Kok entanglement generation scheme) will find natural and straightforward descriptions in the Pauli Fusion model.

\newpage
\section{Proofs}
\subsection{Proof of runnability (Lemma~\ref{lemma:runnability})}\label{app:run}

Every node $v \in V(D_\PFus)$ satisfies either $u \prec v$ or $v \prec u$ for each of its neighbours $u \in V(D)$.
Then either $v$ is a preparation, adjacent to an input, or has a non-empty set of neighbours which precede it in $\preceq$.
Define a function $t: V(D_\PFus) \to \mathbb N$ recursively by setting $t(v) = 1$ for all nodes $v$ which are preparations or which are adjacent to an input, and by recursively defining
\begin{equation}
    t(v)
  \;=\;
    1 \,+\;
    \max \; \Bigl\{ t(u) \;\Big|\; 
      \bigl(u \in V(\mathcal G_D) \,\mathbin\&\, u \prec v\bigr) \text{ or } \\ (u \to v) \in E(D_\PFus) \Bigr\}
\end{equation}
for all other $v \in V(D_\PFus)$.
By construction, we have $t(u) < t(v)$ for $u,v \in V(D_\PFus)$ with an edge $u \to v$.
Furthermore, if for some $u,v \in V(D_\PFus)$ we have $\tau(v) \in \{\rotV {} {\alpha,S,T}, \rotH {} {\alpha,S,T}\}$ where $u \in S \cup T$, it follows that one of the following two cases hold:
\begin{itemize}
\item
  $\delta^+(u) = 0$, and $v \in \bigl(C \setminus \mathcal P) \cup \bigl(\Odd(C_{u,u}) \setminus \{u\}\bigr)$, from which it follows that $u \prec v$ by the conditions which hold of $C_{u,u}$.
  Then $t(u) < t(v)$.
\item  
  $v \in \bigl(C_{\tilde u, \tilde v} \setminus \mathcal P\bigr) \cup \bigl(\Odd(C_{\tilde u, \tilde v}) \setminus \{ \tilde u \}\bigr)$ for some vertices $\tilde u, \tilde v \in V(\mathcal G_D)$ such that $\tilde u \in N^-(\tilde v)$, and it happens that $u$ is an element of $P_{\tilde u, \tilde v}$, the set of vertices which are generated in the decomposition of a higher-degree node with label $\tilde v$ in the ZX diagram $D$, and which lie on a path between $\tilde u$ and $\tilde v$.
  From this it follows that $t(\tilde u) < t(u) < t(\tilde v)$; and as $\tilde v \prec v$ from the conditions on $C_{\tilde u, \tilde v}$, it follows that $t(u) < t(\tilde v) < t(v)$.
\end{itemize}
Then $t$ is a time-ordering of $D_\PFus$, from which Lemma~\ref{lemma:runnability} follows.

\subsection{Proof of determinism (Lemma~\ref{lemma:determinism})}\label{app:det}

 We show, given $s \in \{0,1\}^{\mathcal B}$, how to transform $D_\PFus(s)$ into $D$ using transformations of the ZX-calculus which preserve the semantics (up to a sign).
  We proceed by induction on the Hamming weight $\lVert s \rVert_{\mathrm{H}}$ of $s$.
  If $s = 00\cdots0$, then it is easy to see that $\denote{D_\PFus(s)} = \denote{D}$: all $\pi$-phases depending on bits $s_b$ for $b \in \mathcal B$ either contribute $0$ to a rotation/projection node, or are equivalent to the identity.
  We may then use the ZX calculus to reverse the transformations of \PFcompilation, thereby obtaining the diagram $D$.
  
  Suppose instead that there is some $v \in \mathcal B$ such that $s_v = 1$.
  Let $D_\PFus|_{s_v=b}$ be the AZX diagram obtained from $D_\PFus$ by setting $s_v \gets b$, for any $b \in \{0,1\}$.
  We show how to rewrite $D_\PFus|_{s_v=1}$ to $D_\PFus|_{s_v=0}$ using the ZX calculus, accruing at most a sign of $-1$ in the semantics as we do so, allowing us to prove the equivalence of $D_\PFus(s)$ to $D_\PFus(\tilde s)$ for a string $\tilde s$ that has fewer bits set to $1$ than $s$ does.
  \begin{itemize}
  \item  
    If $\tau(v) \in \bigl\{\projH {}, \projV {}\bigr\}$, there is a set $C_{v,v} \in \mathfrak C$ which is a $v$-corrector at $v$: denote this by $C$, let $\tilde v := v$.
    There is also a node in $D_\PFus$ with a label $w \in \Odd(C)$, such that $\delta^+(w) = 0$ in $\mathcal G_D$, with the same colour as $v$.
    We rewrite the diagram $D_\PFus$ as an AZX diagram, by propagating the phase $s_v \pi$ from $v$ to $w$.
    Denote the resulting diagram by $D'_\PFus$.
  \item
    Otherwise, if $\tau(v) \in \bigl\{ \mergH {}, \mergV {} \bigr\}$, then there is an associated vertex $\tilde v \in V(D_\PFus)$, corresponding to a node-label $\tilde v \in V(D)$ of degree $3$ or greater, and one or more node-labels $\tilde u_1, \tilde u_2, \ldots \in N^-(\tilde v)$ in $\mathcal G_D$, such that $v \in P_{\tilde u_j, \tilde v}$ for each $j \ge 1$.
    For each of these nodes $\tilde u_j$\,, there is a $\tilde u_j$-corrector at $\tilde v$, $C_{\tilde u_j,\tilde v}$.
    We let $C$ be the symmetric difference $C := C_{\tilde u_1, \tilde v} \mathbin\Delta C_{\tilde u_2, \tilde v} \mathbin\Delta \cdots$ of these (just $C_{\tilde u, \tilde v}$ if there is only one such $\tilde u$).
    We also rewrite the diagram $D_\PFus$ as an AZX diagram, by propagating the phase $s_v \pi$ away from the node $v$, against the orientation of the edges, towards the nodes with labels $\tilde u_j$\,, explicitly accumulating the phase $s_v \pi$ on these nodes.
    (This may involve one or application of the bialgebra law between $s_v \pi$ nodes and opposite-coloured nodes of degree $3$, and applications of the spider rule between the nodes $\tilde u_j$ and $s_v \pi$ phase nodes of the same colour.)
    Denote the resulting diagram by $D'_\PFus$.
  \end{itemize}
  In either case, we have a set $C$ which is involved in the correction of vertices, which is involved in annotations on other vertices involving the bit $s_v$.
  Specifically, if $\tilde v = v$ is a projection, then $s_v$ governs a $\pi$-phase contribution for all nodes $w \in \Odd(C) \setminus \{ v \}$, and a change in sign of the phase of all nodes $t \in C \setminus \mathcal P$.
  Otherwise $v$ is a vertex in each of a collection of paths $P_{\tilde u_j,v}$, so that $s_v$ governs (possibly canceling) $\pi$-phase contributions for all nodes $w \in \Odd(C_{\tilde u_j,\tilde v}) \setminus \{ \tilde v \}$ for each $\tilde u_j$, and (possibly canceling) sign-flips of the phase of all nodes $t \in C_{\tilde u_j, \tilde v} \setminus \mathcal P$ for each $\tilde u_j$.
  
  \begin{subequations}%
  \medskip\noindent
  Using the set $C$ constructed as above, we consider the following rewrites on $D_\PFus$:
  \begin{enumerate}
  \item
    For each vertex-label $t \in C \setminus \mathcal P$, by construction the diagram $D_\PFus$ will contain a generator $\rotV t {\alpha,S,T}$ or $\rotH t {\alpha,S,T}$ where $v \in S$.
    We thus rewrite the diagram as an AZX diagram by surrounding $t$ with $s_v \pi$-phase nodes of the opposite colour, and remove $v$ from the set of variables $S$ which may govern a change of sign of the rotation, \eg:
    \begin{align}
      {\input{Mcorr.tikz}} \;\;\mapsto{}&\;\; {\input{corr_s2R.tikz}}
    \end{align}
    Denote the resulting diagram by $D_\PFus''$.
  \item
    For each vertex-label $t \in C \cap \mathcal P$, by construction the diagram $D_\PFus$ will contain a generator $\rotV t {\alpha,\emptyset,T}$ or $\rotH t {\alpha,\emptyset,T}$ where $\alpha$ is an integer multiple of $\pi/2$ (\ie,~$2\alpha/\pi$ is an integer), and $v$ may or may not be an element of $T$.
    \begin{itemize}
    \item
      If $2\alpha/\pi$ is even, then we rewrite the diagram by surrounding $t$ with $s_v \pi$-phase nodes of the opposite colour, without any other changes, \eg: 
      \begin{align}
        {\input{Mcorr_pion2.tikz}} \;\;\mapsto{}&\;\; {\input{corr_s2R_pi.tikz}}
      \end{align}
    \item
      If $2\alpha/\pi$ is odd, then we rewrite the diagram by surrounding $t$ with $s_v \pi$-phase nodes of the opposite colour, and ``toggling'' the membership of $v$ in $T$, \eg:
      \begin{align}
        {\input{Mcorr_pion2.tikz}} \;\;\mapsto{}&\;\; {\input{corr_s2R_pion2.tikz}}
      \end{align}
    \end{itemize}
    Denote the resulting diagram by $D_\PFus'''$.
  \item
    For any $s_v\pi$ phase produced in the previous steps which is adjacent to a node $t \in C$, propagate the phase node away from $t$ (either consistently with the orientation of the edges, or consistently against the orientation) until it is adjacent to a node with label $w \in \Odd(C)$ of the same colour, or more generally forms  part of a chain of $s_v \pi$ phase nodes of which one end is adjacent to a node with label $w \in \Odd(C)$ of the same colour.
    To do this, it may be necessary to propagate two nodes with phase $s_v \pi$ of opposite colour past one another in opposite directions: this induces at most a change of sign due to anticommutation of $X$ and $Z$.
    Denote the resulting diagram by $\tilde D_\PFus$.

\item
    For each $w \in \Odd(C)$, $\tilde D_\PFus$ contains a generator $\rotH w {\alpha,S,T}$ or a generator $\rotV w {\alpha,S,T}$.
    \begin{itemize}
    \item  
      If $\alpha$ is an odd multiple of $\pi/2$ and $w \in C$, then by construction there will be an even number of $s_v \pi$ phase nodes, either adjacent to $w$ it or more generally in one or two chains which are adjacent to $w$, and we will have $v \notin T$.
      We may then absorb all of these phases into $w$ without modifying the generator at $w$.
    \item  
      If $\alpha$ is an odd multiple of $\pi/2$ and $w \notin C$, or if $\alpha$ is not an odd multiple of $\pi/2$, then by construction there will be an odd number of $s_v \pi$ phase nodes, either adjacent to $w$ it or more generally in one or two chains which are adjacent to $w$, and we will have $v \in T$.
      We may then absorb all of these phases into $w$ if we replace  $\rotH w {\alpha,S,T}$ with $\rotH w {\alpha,S,T\mathbin\Delta\{v\}}$ or replace $\rotV w {\alpha,S,T}$ with $\rotV w {\alpha,S,T\mathbin\Delta\{v\}}$.
    \end{itemize}%
  \end{enumerate}%
  \end{subequations}%
  This sequence of rewrites has the effect of removing all instances of $s_v$ from the diagram, \ie,~it removes $v$ from all sets which modify the phases of rotations, without affecting any other influences on the phases and (at the end) without there being any other change to to the structure of the diagram from $D_\PFus$.
  Thus the diagram is equivalent to $D_\PFus|_{s_v = 0}$, while incurring a change in semantics by at most a sign.
  From this it follows that $\denote{D_\PFus(s)} = \pm\denote{D_\PFus(\tilde s)}$.
  
  By induction, it follows that $\denote{D_\PFus(s)} = \pm \denote{D_\PFus(00\cdots0)}$ for any $s \in \{0,1\}^{\mathcal B}$, proving Lemma~\ref{lemma:determinism}.

\subsection{Proof that \PFflowfinding\ halts in polynomial time (Lemma~\ref{lemma:flowFindingRuntime})}\label{app:halt}

As the first operation of the loop is to set $\delta M := \emptyset$, the algorithm will terminate in any loop where Step~3c\:\!\ref{step:mark_an_element} is not executed at least once, in which an element of $V(D) \setminus M$ is added to $\delta M$.
If this step is run, then Step~\ref{step:mark_vertices} increases the size of $M$, decreasing the set of elements which may be added to $\delta M$ in subsequent loops.
It follows that the loop is executed at most $n$ times.

In each iteration, the operations performed consist largely of tests for membership in sets, or constructions of intersections, subtractions, unions, or products of sets, each of which can be performed in polynomial time.
The step whose cost is least obvious is the computation of the set $R$, whose cost we describe below.

Given a vertex $u$, finding whether there is a corresponding set $C_u$ amounts to solving a system of equations in $\mathbb F_2$,
\begin{equation}
  \label{eqn:testCorrectabilityGF2}
  \mathbf A_{[M]}\, \mathbf x = \mathbf e_u\;,  
\end{equation}
where $\mathbf e_u \in \mathbb F_{2}^{\:\!\smash{V(D) \setminus M}}$ is the characteristic vector of the set $\{u\} \subseteq V(D) \setminus M$, and $\mathbf A_{[M]}$ denotes the submatrix of the adjacency matrix $\mathbf A$ of $D$ whose rows are indexed by $V(D) \setminus M$ and whose columns are indexed by $M \cup \mathcal P$.
Each column of $\mathbf A_{[M]}$ is the characteristic vector for the set of vertices which are adjacent to an element of $t \in M \cup \mathcal P$; and which are not yet marked.
Then $\mathbf A_{[M]} \, \mathbf x$ represents the set of unmarked vertices which are adjacent to an odd number of some set of vertices $X \subset M \cup \mathcal P$, whose elements are indicated by the non-zero coefficients of $\mathbf x$.
Determining whether Eqn.~\eqref{eqn:testCorrectabilityGF2} has solutions can be performed efficiently, as can producing one such solution $\mathbf c_u$, which then represents a characteristic vector of a corrector set $C_u$.

Thus \PFflowfinding\ halts in polynomial time for any graph-like ZX diagram $D$, proving Lemma~\ref{lemma:flowFindingRuntime}.

\subsection{Proof of correctness of \PFflowfinding\ algorithm (Lemma~\ref{lemma:flowFindingCompleteness})}\label{app:suc}

  Let $(\tpreceq, \tilde f, \tilde{\mathfrak C})$ be a PF-Flow for $D$.
  In particular, any two vertices which are adjacent in $\mathcal G_D$ are comparable in $\preceq$; there is a $v$-corrector at $v$ ($\tilde C_{v,v} \in \tilde{\mathfrak C}$) for any vertex $v \in V(D)$ whose neighbours all precede it; and there is a $u$-corrector at $v$ ($\tilde C_{u,v} \in \mathfrak C$) for any adjacent vertices $u \preceq v$ from $D$ such that $u \ne \tilde f(v)$.

  Consider a maximal element $\omega \in V(D)$ with respect to $\tpreceq$.
  As any $t \in V(\mathcal G_D) \setminus \{\omega\}$ for which $\omega \tpreceq t$ can neither be an element of $V(D)$ nor $\mathcal I$, it would follow that $t \in \mathcal O$.
  Thus for $C$ any $u$-corrector set at $\omega$, whether $u = \omega$ or $u \in N(\omega)$, must have the properties that $C \setminus \mathcal P \subseteq \mathcal O$ and $\Odd(C) \setminus \{w\} \subseteq \mathcal O$.
  On the first iteration of the loop, we have $M = \mathcal O$, so that the corrector sets $\tilde C_{u,\omega}$ satisfy the conditions in the construction of the set $R$ (though the algorithm may find different corrector sets $C_u$), for all $u \in N(\omega) \setminus M$ excepting possibly $u = \tilde f(\omega)$.  
  (If $\omega$ has no neighbours in $\mathcal O$, a corrector set $\tilde C_{\omega,\omega}$ also exists, so that $\omega \in R$ at this stage as well.
  Again, the corrector set $C_\omega$ may differ from $\tilde C_\omega$.)
  In the iteration through $v \in V(D) \setminus M$ in this first loop, we will then find $\lvert F_\omega \rvert \le 1$, and add $\omega$ to $\delta M$, so that it will become a maximal element of the relation $\preceq$.
  If $F_\omega$ is non-empty, it will contain $\tilde f(\omega)$, so that we will have $f(\omega) = \tilde f(\omega)$ in this case.
  
  We may show by induction that \PFflowfinding\ will construct a PF-Flow ${(\preceq, f, \mathfrak C)}$, by noting that in each iteration there will be a maximal element of $\tpreceq$ in $V(D)$ subject to not yet having been marked in a previous iteration, which by a similar analysis will be added to $\delta M$ in that iteration, and that the data $(\preceq, f, \mathfrak C)$ satisfy the properties of a PF-Flow by construction.  
  This proves Lemma~\ref{lemma:flowFindingCompleteness}.

\end{document}

%% file: zx.tikzstyles
\definecolor{zx_red}{RGB}{227, 145, 145}
\definecolor{zx_green}{RGB}{216, 248, 216}

\tikzstyle{box}=[shape=rectangle, text height=1.5ex, text depth=0.25ex, yshift=0.5mm, fill=white, draw=black, minimum height=5mm, yshift=-0.5mm, minimum width=5mm, font={\small}]
\tikzstyle{Z dot}=[inner sep=0mm, minimum size=2mm, shape=circle, draw=black, fill=zx_green, inner sep=2pt]
\tikzstyle{Z phase dot}=[minimum size=5mm, font={\footnotesize\boldmath}, shape=rectangle, rounded corners=2mm, inner sep=0.2mm, outer sep=-2mm, scale=0.8, tikzit shape=circle, draw=black, fill=zx_green, tikzit draw=blue]
\tikzstyle{X dot}=[Z dot, shape=circle, draw=black, fill=zx_red]
\tikzstyle{X phase dot}=[Z phase dot, tikzit shape=circle, tikzit draw=blue, fill=zx_red, font={\footnotesize\boldmath}, inner sep=2pt]
\tikzstyle{hadamard}=[fill=yellow, draw=black, shape=rectangle, inner sep=0.6mm, minimum height=1.5mm, minimum width=1.5mm]
\tikzstyle{vertex}=[inner sep=0mm, minimum size=1.5mm, shape=circle, draw=black, fill=black]
\tikzstyle{sq}=[inner sep=0mm, minimum size=1.5mm, shape=circle, draw=black, fill=white]
\tikzstyle{vertex set}=[inner sep=0mm, minimum size=1.5mm, shape=circle, draw=black, fill=white,font={\footnotesize\boldmath}]

\tikzstyle{hadamard edge}=[-,color=blue,dashed,dash pattern=on 2pt off 0.7pt]
\tikzstyle{brace edge}=[-,tikzit draw=blue,decorate,decoration={brace,amplitude=1mm,raise=-1mm}]

%% file: figures/ZX-VProj0.tikz
\begin{tikzpicture}[scale=.8]
	\begin{pgfonlayer}{nodelayer}
		\node [style=none] (3) at (-1.25, 0) {};
		\node [style=Z dot, minimum width=7.5pt, minimum height=7.5pt, inner sep=0pt] (4) at (0, 0) {};
	\end{pgfonlayer}
	\begin{pgfonlayer}{edgelayer}
		\draw (3.center) to (4);
	\end{pgfonlayer}
\end{tikzpicture}

%% file: figures/ZX-VProj1.tikz
\begin{tikzpicture}[scale=.8]
	\begin{pgfonlayer}{nodelayer}
		\node [style=none] (3) at (-1.25, 0) {};
		\node [style=Z dot, minimum width=7.5pt, minimum height=7.5pt, inner sep=0pt] (4) at (0, 0) {$\scriptstyle\pi$};
	\end{pgfonlayer}
	\begin{pgfonlayer}{edgelayer}
		\draw (3.center) to (4);
	\end{pgfonlayer}
\end{tikzpicture}

%% file: figures/ZX-VMerge0.tikz
\begin{tikzpicture}[scale=0.5]
	\begin{pgfonlayer}{nodelayer}
		\node [style=Z dot, minimum width=7.5pt, minimum height=7.5pt] (0) at (0, 0) {};
		\node [style=none] (1) at (1.5, 0) {};
		\node [style=none] (2) at (-2.25, 1) {};
		\node [style=none] (3) at (-2.25, -1) {};
	\end{pgfonlayer}
	\begin{pgfonlayer}{edgelayer}
		\draw (1.center) to (0);
		\draw [out=135,in=0] (0) to (2.center);
		\draw [out=-135,in=0] (0) to (3.center);
	\end{pgfonlayer}
\end{tikzpicture}

%% file: figures/ZX-VMerge1.tikz
\begin{tikzpicture}[scale=0.5]
	\begin{pgfonlayer}{nodelayer}
		\node [style=Z dot, minimum width=7.5pt, minimum height=7.5pt] (0) at (0, 0) {};
		\node [style=none] (1) at (1.5, 0) {};
		\node [style=none] (2) at (-2.25, 1) {};
		\node [style=none] (3) at (-2.25, -1) {};
                \node [style=X dot, inner sep=0pt, minimum width=7.5pt, minimum height=7.5pt] (4) at (-1.25, -0.875) {$\scriptstyle\pi$};
	\end{pgfonlayer}
	\begin{pgfonlayer}{edgelayer}
		\draw (1.center) to (0);
		\draw [out=135,in=0] (0) to (2.center);
		\draw [out=-135,in=0] (0) to (3.center);
	\end{pgfonlayer}
\end{tikzpicture}

%% file: figures/ZX-HProj0.tikz
\begin{tikzpicture}[scale=.8]
	\begin{pgfonlayer}{nodelayer}
		\node [style=none] (3) at (-1.25, 0) {};
		\node [style=X dot, minimum width=7.5pt, minimum height=7.5pt, inner sep=0pt] (4) at (0, 0) {};
	\end{pgfonlayer}
	\begin{pgfonlayer}{edgelayer}
		\draw (3.center) to (4);
	\end{pgfonlayer}
\end{tikzpicture}

%% file: figures/ZX-HProj1.tikz
\begin{tikzpicture}[scale=.8]
	\begin{pgfonlayer}{nodelayer}
		\node [style=none] (3) at (-1.25, 0) {};
		\node [style=X dot, minimum width=7.5pt, minimum height=7.5pt, inner sep=0pt] (4) at (0, 0) {$\scriptstyle\pi$};
	\end{pgfonlayer}
	\begin{pgfonlayer}{edgelayer}
		\draw (3.center) to (4);
	\end{pgfonlayer}
\end{tikzpicture}

%% file: figures/ZX-HMerge0.tikz
\begin{tikzpicture}[scale=0.5]
	\begin{pgfonlayer}{nodelayer}
		\node [style=X dot, minimum width=7.5pt, minimum height=7.5pt] (0) at (0, 0) {};
		\node [style=none] (1) at (1.5, 0) {};
		\node [style=none] (2) at (-2.25, 1) {};
		\node [style=none] (3) at (-2.25, -1) {};
	\end{pgfonlayer}
	\begin{pgfonlayer}{edgelayer}
		\draw (1.center) to (0);
		\draw [out=135,in=0] (0) to (2.center);
		\draw [out=-135,in=0] (0) to (3.center);
	\end{pgfonlayer}
\end{tikzpicture}

%% file: figures/ZX-HMerge1.tikz
\begin{tikzpicture}[scale=0.5]
	\begin{pgfonlayer}{nodelayer}
		\node [style=X dot, minimum width=7.5pt, minimum height=7.5pt] (0) at (0, 0) {};
		\node [style=none] (1) at (1.5, 0) {};
		\node [style=none] (2) at (-2.25, 1) {};
		\node [style=none] (3) at (-2.25, -1) {};
                \node [style=Z dot, inner sep=0pt, minimum width=7.5pt, minimum height=7.5pt] (4) at (-1.25, -0.875) {$\scriptstyle\pi$};
	\end{pgfonlayer}
	\begin{pgfonlayer}{edgelayer}
		\draw (1.center) to (0);
		\draw [out=135,in=0] (0) to (2.center);
		\draw [out=-135,in=0] (0) to (3.center);
	\end{pgfonlayer}
\end{tikzpicture}

%% file: figures/ZX-VRot.tikz
\begin{tikzpicture}[scale=0.5]
	\begin{pgfonlayer}{nodelayer}
		\node [style=Z dot, minimum width=7.5pt, minimum height=7.5pt, inner sep=0pt] (0) at (0, 0) {$\scriptstyle\alpha$};
		\node [style=none] (1) at (-1.5, 0) {};
		\node [style=none] (3) at (1.5, 0) {};
	\end{pgfonlayer}
	\begin{pgfonlayer}{edgelayer}
		\draw (1.center) to (0);
		\draw (0) to (3.center);
	\end{pgfonlayer}
\end{tikzpicture}

%% file: figures/ZX-HRot.tikz
\begin{tikzpicture}[scale=0.5]
	\begin{pgfonlayer}{nodelayer}
		\node [style=X dot, minimum width=7.5pt, minimum height=7.5pt, inner sep=0pt] (0) at (0, 0) {$\scriptstyle\alpha$};
		\node [style=none] (1) at (-1.5, 0) {};
		\node [style=none] (3) at (1.5, 0) {};
	\end{pgfonlayer}
	\begin{pgfonlayer}{edgelayer}
		\draw (1.center) to (0);
		\draw (0) to (3.center);
	\end{pgfonlayer}
\end{tikzpicture}

%% file: figures/VSplit.tikz
\begin{tikzpicture}
	\begin{pgfonlayer}{nodelayer}
		\node [style=Z dot] (0) at (-0.25, 0) {};
		\node [style=none] (1) at (-1.75, 0) {};
		\node [style=none] (2) at (1.5, 1) {};
		\node [style=none] (3) at (1.5, -1) {};
		\node [style=none] (5) at (-0.75, 0.5) {$u$};
	\end{pgfonlayer}
	\begin{pgfonlayer}{edgelayer}
		\draw (1.center) to (0);
		\draw [bend left] (0) to (2.center);
		\draw [bend right] (0) to (3.center);
	\end{pgfonlayer}
\end{tikzpicture}

%% file: figures/VMerge.tikz
\begin{tikzpicture}
	\begin{pgfonlayer}{nodelayer}
		\node [style=Z dot] (0) at (0, 0) {};
		\node [style=none] (1) at (1.5, 0) {};
		\node [style=none] (2) at (-1.75, 1) {};
		\node [style=none] (3) at (-1.75, -1) {};
		\node [style=X phase dot] (4) at (-0.75, -0.75) {$s_u\pi$};
		\node [style=none] (5) at (0.5, 0.5) {$u$};
	\end{pgfonlayer}
	\begin{pgfonlayer}{edgelayer}
		\draw (1.center) to (0);
		\draw [bend right] (0) to (2.center);
		\draw [bend left] (0) to (3.center);
	\end{pgfonlayer}
\end{tikzpicture}

%% file: figures/VRot.tikz
\begin{tikzpicture}
	\begin{pgfonlayer}{nodelayer}
		\node [style=Z phase dot] (0) at (0, 0) {$\;\Theta(\alpha,S,T)\Big.\;$};
		\node [style=none] (1) at (-2, 0) {};
		\node [style=none] (3) at (2, 0) {};
		\node [style=none] (5) at ($(0.east) + (.5, 0.75)$) {$u$};
	\end{pgfonlayer}
	\begin{pgfonlayer}{edgelayer}
		\draw (1.center) to (0);
		\draw (0) to (3.center);
	\end{pgfonlayer}
\end{tikzpicture}

%% file: figures/HSplit.tikz
\begin{tikzpicture}
	\begin{pgfonlayer}{nodelayer}
		\node [style=X dot] (0) at (-0.25, 0) {};
		\node [style=none] (1) at (-1.75, 0) {};
		\node [style=none] (2) at (1.5, 1) {};
		\node [style=none] (3) at (1.5, -1) {};
		\node [style=none] (5) at (-0.75, 0.5) {$u$};
	\end{pgfonlayer}
	\begin{pgfonlayer}{edgelayer}
		\draw (1.center) to (0);
		\draw [bend left] (0) to (2.center);
		\draw [bend right] (0) to (3.center);
	\end{pgfonlayer}
\end{tikzpicture}

%% file: figures/HMerge.tikz
\begin{tikzpicture}
	\begin{pgfonlayer}{nodelayer}
		\node [style=X dot] (0) at (0, 0) {};
		\node [style=none] (1) at (1.5, 0) {};
		\node [style=none] (2) at (-1.75, 1) {};
		\node [style=none] (3) at (-1.75, -1) {};
		\node [style=Z phase dot] (4) at (-0.75, -0.75) {$s_u\pi$};
		\node [style=none] (5) at (0.5, 0.5) {$u$};
	\end{pgfonlayer}
	\begin{pgfonlayer}{edgelayer}
		\draw (1.center) to (0);
		\draw [bend right] (0) to (2.center);
		\draw [bend left] (0) to (3.center);
	\end{pgfonlayer}
\end{tikzpicture}

%% file: figures/HRot.tikz
\begin{tikzpicture}
	\begin{pgfonlayer}{nodelayer}
		\node [style=X phase dot] (0) at (0, 0) {$\;\Theta(\alpha,S,T)\Big.\;$};
		\node [style=none] (1) at (-2, 0) {};
		\node [style=none] (3) at (2, 0) {};
		\node [style=none] (5) at ($(0.east) + (.5, 0.75)$) {$u$};
	\end{pgfonlayer}
	\begin{pgfonlayer}{edgelayer}
		\draw (1.center) to (0);
		\draw (0) to (3.center);
	\end{pgfonlayer}
\end{tikzpicture}

%% file: figures/VInit.tikz
\begin{tikzpicture}
	\begin{pgfonlayer}{nodelayer}
		\node [style=none] (3) at (0.75, 0) {};
		\node [style=Z dot] (4) at (-0.25, 0) {};
		\node [style=none] (5) at (0, 0.5) {$u$};
	\end{pgfonlayer}
	\begin{pgfonlayer}{edgelayer}
		\draw (3.center) to (4);
	\end{pgfonlayer}
\end{tikzpicture}

%% file: figures/HInit.tikz
\begin{tikzpicture}
	\begin{pgfonlayer}{nodelayer}
		\node [style=none] (3) at (0.75, 0) {};
		\node [style=X dot] (4) at (-0.25, 0) {};
		\node [style=none] (5) at (0, 0.5) {$u$};
	\end{pgfonlayer}
	\begin{pgfonlayer}{edgelayer}
		\draw (3.center) to (4);
	\end{pgfonlayer}
\end{tikzpicture}

%% file: figures/H.tikz
\begin{tikzpicture}
	\begin{pgfonlayer}{nodelayer}
		\node [style=none] (0) at (-1, 0) {};
		\node [style=none] (1) at (1, 0) {};
		\node [style=hadamard, label=north east: \!\!$u$] (2) at (0, 0) {\phantom x};
	\end{pgfonlayer}
	\begin{pgfonlayer}{edgelayer}
		\draw (0.center) to (1.center);
	\end{pgfonlayer}
\end{tikzpicture}

%% file: figures/VProj.tikz
\begin{tikzpicture}
	\begin{pgfonlayer}{nodelayer}
		\node [style=none] (3) at (-1.25, 0) {};
		\node [style=Z phase dot] (4) at (0, 0) {$\;s_u\pi\;$};
		\node [style=none] (5) at ($(4.east) + (0.75, 0.5)$) {$u$};
	\end{pgfonlayer}
	\begin{pgfonlayer}{edgelayer}
		\draw (3.center) to (4);
	\end{pgfonlayer}
\end{tikzpicture}

%% file: figures/HProj.tikz
\begin{tikzpicture}
	\begin{pgfonlayer}{nodelayer}
		\node [style=none] (3) at (-1.25, 0) {};
		\node [style=X phase dot] (4) at (0, 0) {$\;s_u\pi\;$};
		\node [style=none] (5) at (0.75, 0.5) {$u$};
	\end{pgfonlayer}
	\begin{pgfonlayer}{edgelayer}
		\draw (3.center) to (4);
	\end{pgfonlayer}
\end{tikzpicture}

%% file: figures/Swap.tikz
\begin{tikzpicture}
	\begin{pgfonlayer}{nodelayer}
		\node [style=none] (0) at (-1, 0.75) {};
		\node [style=none] (1) at (1, -0.75) {};
		\node [style=none] (2) at (1, 0.75) {};
		\node [style=none] (3) at (-1, -0.75) {};
	\end{pgfonlayer}
	\begin{pgfonlayer}{edgelayer}
		\draw [in=180, out=0] (0.center) to (1.center) node [midway] {}; 
		\draw [in=180, out=0] (3.center) to (2.center);
	\end{pgfonlayer}
\end{tikzpicture}

%% file: figures/CNONSENSE.tikz
\begin{tikzpicture}[scale=0.75]
	\begin{pgfonlayer}{nodelayer}
		\node [style=Z dot] (0) at (1.75, 1) {};
		\node [style=none] (1) at (3.25, 1) {};
		\node [style=none] (2) at (0, 2) {};
		\node [style=none] (3) at (0, 0) {};
             	\node [style=X phase dot] (4) at (0.875, 0.125) {$s_w\pi$};
		\node [style=none] (5) at (2.25, 1.5) {\footnotesize $w$};
		\node [style=none] (6) at (-1.125, -1) {\footnotesize $v$};
		\node [style=none] (7) at (-5.5, -1) {};
		\node [style=X dot] (8) at (-1.75, -1) {};
		\node [style=none] (9) at (0, 0) {};
		\node [style=none] (10) at (0, -2) {};
		\node [style=none] (11) at (-5.5, 2) {};
		\node [style=none] (12) at (3.25, -2) {};
		\node [style=X phase dot] (13) at (-4.25, -1) {$s_w\pi$};
		\node [style=none] (6') at ($(13.east) + (0.75, 0.5)$) {\footnotesize $u$};
	\end{pgfonlayer}
	\begin{pgfonlayer}{edgelayer}
		\draw (1.center) to (0);
		\draw [bend right] (0) to (2.center);
		\draw [bend left] (0) to (3.center);
		\draw (7.center) to (8);
		\draw [bend left] (8) to (9.center);
		\draw [bend right] (8) to (10.center);
		\draw (11.center) to (2.center);
		\draw [loop] (4.center) to ();
		\draw [loop] (13.center) to ();
		\draw (10.center) to (12.center);
	\end{pgfonlayer}
\end{tikzpicture}

%% file: figures/EX_0.tikz
\begin{tikzpicture}
	\begin{pgfonlayer}{nodelayer}
		\node [style=none] (0) at (-4, 0) {};
		\node [style=X phase dot] (1) at (-3, 0) {$\alpha$};
		\node [style=Z phase dot] (2) at (-1.5, -1) {$\frac \pi 2$};
		\node [style=X dot] (3) at (0, 0) {};
		\node [style=Z phase dot] (4) at (1.25, 1.25) {$\beta$};
		\node [style=Z dot] (5) at (2.25, -1) {};
		\node [style=X dot] (6) at (0.25, -2) {};
		\node [style=none] (7) at (3, 1.25) {};
		\node [style=none] (8) at (3.25, -1) {};
		\node [style=none] (9) at (3, -2.75) {};
		\node [style=none] (10) at (-3, 0.75) {\tiny $u_1$};
		\node [style=none] (11) at (1.75, 1.75) {\tiny $u_5$};
		\node [style=none] (12) at (-0.5, 0.5) {\tiny $u_3$};
		\node [style=none] (13) at (-1.5, -1.75) {\tiny $u_2$};
		\node [style=none] (14) at (0.25, -1.5) {\tiny $u_4$};
		\node [style=none] (15) at (2.5, -0.5) {\tiny $u_6$};
	\end{pgfonlayer}
	\begin{pgfonlayer}{edgelayer}
		\draw (0.center) to (1);
		\draw [bend left, looseness=1.25] (1) to (4);
		\draw [bend right, looseness=0.75] (1) to (2);
		\draw [bend left] (2) to (3);
		\draw [bend right] (2) to (6);
		\draw (3) to (4);
		\draw (4) to (7.center);
		\draw [bend left] (3) to (5);
		\draw [bend right=15] (6) to (9.center);
		\draw [bend right, looseness=1.25] (6) to (5);
		\draw (5) to (8.center);
	\end{pgfonlayer}
\end{tikzpicture}

%% file: figures/EX_0_OpenGraph.tikz
\begin{tikzpicture}
	\begin{pgfonlayer}{nodelayer}
		\node [style=sq] (0) at (-4, 0) {};
		\node [style=sq] (1) at (-3, 0) {};
		\node [style=vertex] (2) at (-1.5, -1) {};
		\node [style=vertex] (3) at (0, 0) {};
		\node [style=sq] (4) at (1.25, 1.25) {};
		\node [style=vertex] (5) at (2.25, -1) {};
		\node [style=vertex] (6) at (-0.25, -2.5) {};
		\node [style=sq] (7) at (3, 1.25) {};
		\node [style=sq] (8) at (3.25, -1) {};
		\node [style=sq] (9) at (2.75, -2.75) {};
		\node [style=none] (10) at (-3, 0.5) {\tiny $u_1$};
		\node [style=none] (11) at (1.75, 1.75) {\tiny $u_5$};
		\node [style=none] (12) at (-0.5, 0.25) {\tiny $u_3$};
		\node [style=none] (13) at (-1.5, -1.75) {\tiny $u_2$};
		\node [style=none] (14) at (-0.125, -2) {\tiny $u_4$};
		\node [style=none] (15) at (2.5, -0.5) {\tiny $u_6$};
		\node [style=none] (16) at (-4.5, 0) {\tiny $i_1$};
		\node [style=none] (17) at (3.5, 1.25) {\tiny $o_1$};
		\node [style=none] (18) at (3.75, -1) {\tiny $o_2$};
		\node [style=none] (19) at (3.25, -2.75) {\tiny $o_3$};
		\node [style=none] (21) at (-0.75, -1) {};
	\end{pgfonlayer}
	\begin{pgfonlayer}{edgelayer}
		\draw (0.center) to (1);
		\draw (1) to (4);
		\draw (1) to (2);
		\draw (2) to (3);
		\draw (2) to (6);
		\draw (3) to (4);
		\draw (4) to (7.center);
		\draw (3) to (5);
		\draw (6) to (9.center);
		\draw (6) to (5);
		\draw (5) to (8.center);
		\draw [in=90, out=25] (2) to (21.center);
		\draw [in=-90, out=-25] (2) to (21.center);
	\end{pgfonlayer}
\end{tikzpicture}

%% file: figures/first_Spider.tikz
\begin{tikzpicture}[xscale=-0.9,yscale=0.9]
	\begin{pgfonlayer}{nodelayer}
		\node [style=Z phase dot] (0) at (0, 0) {$\boldsymbol\alpha$};
		\node [style=none] (1) at (-2, 1.5) {};
		\node [style=none] (2) at (-2, .75) {};
		\node [style=none] (3) at (-2, -1) {};
		\node [style=none, inner sep=3pt] (4) at (2, 0) {\footnotesize $i_{\smash j}$};
                \node [style=none] (10) at (-1.5, 0.125) {\footnotesize $\vdots$};
		\node [style=none] (6) at (0.25, 0.75) {\footnotesize $v$};
	\end{pgfonlayer}
	\begin{pgfonlayer}{edgelayer}
		\draw [in=135, out=0] (1.center) to (0);
		\draw [in=150, out=0] (2.center) to (0);
		\draw [in=-135, out=0] (3.center) to (0);
		\draw (0) to (4.east);
	\end{pgfonlayer}
\end{tikzpicture}

%% file: figures/first_Spider_xform.tikz
\begin{tikzpicture}[xscale=-0.9,yscale=0.9]
	\begin{pgfonlayer}{nodelayer}
		\node [style=Z phase dot] (0) at (0, 0) {$\boldsymbol\alpha$};
		\node [style=none] (1) at (-2, 1.5) {};
		\node [style=none] (2) at (-2, .75) {};
		\node [style=none] (3) at (-2, -1) {};
                \node [style=none] (10) at (-1.5, 0.125) {\footnotesize $\vdots$};
		\node [style=none] (6) at (0.25, 0.75) {\footnotesize $v$};
		\node [style=X phase dot] (4) at ($(0) + (2,0)$) {};
		\node [style=none] (4l) at ($(4.east) + (-0.75,0.5)$) {\footnotesize $i_{\smash j}$};
	\end{pgfonlayer}
	\begin{pgfonlayer}{edgelayer}
		\draw [in=135, out=0] (1.center) to (0);
		\draw [in=150, out=0] (2.center) to (0);
		\draw [in=-135, out=0] (3.center) to (0);
		\draw (0) to (4) to ++(1,0);
	\end{pgfonlayer}
\end{tikzpicture}

%% file: figures/final_Spider.tikz
\begin{tikzpicture}[scale=0.9]
	\begin{pgfonlayer}{nodelayer}
		\node [style=X phase dot] (0) at (0, 0) {$\boldsymbol\alpha$};
		\node [style=none] (1) at (-2, 1.5) {};
		\node [style=none] (2) at (-2, .75) {};
		\node [style=none] (3) at (-2, -1) {};
		\node [style=none, inner sep=3pt] (4) at (2, 0) {\footnotesize $o_{\smash j}$};
                \node [style=none] (10) at (-1.5, 0.125) {\footnotesize $\vdots$};
		\node [style=none] (6) at (0.25, 0.75) {\footnotesize $v$};
	\end{pgfonlayer}
	\begin{pgfonlayer}{edgelayer}
		\draw [in=135, out=0] (1.center) to (0);
		\draw [in=150, out=0] (2.center) to (0);
		\draw [in=-135, out=0] (3.center) to (0);
		\draw (0) to (4.west);
	\end{pgfonlayer}
\end{tikzpicture}

%% file: figures/final_Spider_xform.tikz
\begin{tikzpicture}[scale=0.9]
	\begin{pgfonlayer}{nodelayer}
		\node [style=X phase dot] (0) at (0, 0) {$\boldsymbol\alpha$};
		\node [style=none] (1) at (-2, 1.5) {};
		\node [style=none] (2) at (-2, .75) {};
		\node [style=none] (3) at (-2, -1) {};
                \node [style=none] (10) at (-1.5, 0.125) {\footnotesize $\vdots$};
		\node [style=none] (6) at (0.25, 0.75) {\footnotesize $v$};
		\node [style=Z phase dot] (4) at ($(0) + (2,0)$) {};
		\node [style=none] (4l) at ($(4.east) + (0.75,0.5)$) {\footnotesize $o_{\smash j}$};
	\end{pgfonlayer}
	\begin{pgfonlayer}{edgelayer}
		\draw [in=135, out=0] (1.center) to (0);
		\draw [in=150, out=0] (2.center) to (0);
		\draw [in=-135, out=0] (3.center) to (0);
		\draw (0) to (4) to ++(1,0);
	\end{pgfonlayer}
\end{tikzpicture}

%% file: figures/Spider_comp.tikz
\begin{tikzpicture}[scale=0.9]
	\begin{pgfonlayer}{nodelayer}
		\node [style=X phase dot] (0) at (0, 0) {$\boldsymbol\alpha$};
		\node [style=none] (1) at (-3, 2) {};
		\node [style=none] (2) at (-3, 1) {};
		\node [style=none] (3) at (-3, -2) {};
		\node [style=none] (4) at (3, 1.5) {};
		\node [style=none] (5) at (3, -1.5) {};
		\node [style=none] (6) at (0.25, 0.75) {\footnotesize $v$};
		\node [style=none] (7) at (-3.75, 2) {\footnotesize $f(v)$};
		\node [style=none] (8) at (-3.75, 1) {\footnotesize $u_1$};
		\node [style=none] (9) at (-3.75, -2) {\footnotesize $u_k$};
		\node [style=none] (10) at (-2, -0.25) {\footnotesize $\vdots$};
		\node [style=none] (11) at (2, 0) {\footnotesize $\vdots$};
	\end{pgfonlayer}
	\begin{pgfonlayer}{edgelayer}
		\draw [in=135, out=0] (1.center) to (0);
		\draw [in=150, out=0] (2.center) to (0);
		\draw [in=-135, out=0] (3.center) to (0);
		\draw [in=180, out=45] (0) to (4.center);
		\draw [in=180, out=-45] (0) to (5.center);
		\draw [loop] (0) to ();
	\end{pgfonlayer}
\end{tikzpicture}

%% file: figures/Spider_decomp.tikz
\begin{tikzpicture}[scale=0.9]
	\begin{pgfonlayer}{nodelayer}
		\node [style=X phase dot] (0) at (0, 0) {$\boldsymbol\alpha$};
		\node [style=none] (6) at (0.25, 0.75) {\footnotesize $v$};
		\node [style=X dot] (7) at (1.5, 0) {};
		\node [style=X dot] (8) at (1.5, 0) {};
		\node [style=X dot] (9) at (3, 1) {};
		\node [style=X dot] (10) at (3, -1) {};
		\node [style=none] (11) at (4.25, 1.75) {};
		\node [style=none] (12) at (4.25, 0.5) {};
		\node [style=none] (13) at (4.25, -0.25) {};
		\node [style=none] (14) at (4.25, -1.75) {};
		\node [style=X dot] (15) at (-1.5, 0) {};
		\node [style=X dot] (16) at (-4, 1) {};
		\node [style=X dot] (17) at (-4, -1) {};
		\node [style=none] (18) at (-6.25, -2) {};
		\node [style=none] (19) at (-6.25, -0.5) {};
		\node [style=none] (20) at (-6.25, 0.5) {};
		\node [style=none] (21) at (-6.25, 2) {};
		\node [style=none] (22) at (-3.75, 1.5) {\footnotesize $v_2$};
		\node [style=none] (23) at (-3.75, -1.5) {\footnotesize $v_3$};
		\node [style=none] (24) at (-1.25, 0.5) {\footnotesize $v_1$};
		\node [style=none] (25) at (-6.75, 2.25) {\footnotesize $\cdots$};
		\node [style=none] (26) at (-6, 1) {};
		\node [style=none] (27) at (-6, -1) {};
		\node [style=none] (28) at (-6.75, -2.25) {\footnotesize $\cdots$};
		\node [style=none] (29) at (1.25, 0.5) {};
		\node [style=none] (30) at (2.5, 1.5) {};
		\node [style=none] (31) at (2.75, -1.5) {};
		\node [style=Z phase dot] (32) at (-2.75, -0.75) {\footnotesize $\,s_{v_1}\pi\,$};
		\node [style=Z phase dot] (33) at (-5.25, -1.75) {\footnotesize $\,s_{v_3}\pi\,$};
		\node [style=Z phase dot] (34) at (-5.25, 0.625) {\footnotesize $\,s_{v_2}\pi\,$};
		\node [style=X dot] (35) at (-8, 2.75) {};
		\node [style=X dot] (36) at (-8, -2.75) {};
		\node [style=none] (37) at (-10.125, 3.5) {};
		\node [style=none] (38) at (-10.125, 2) {};
		\node [style=none] (39) at (-10.125, -2) {};
		\node [style=none] (40) at (-10.125, -3.5) {};
		\node [style=none] (41) at (-7.25, -2.5) {};
		\node [style=none] (42) at (-7.25, 2.5) {};
		\node [style=Z phase dot] (43) at (-9, -3.5) {\footnotesize $\,s_{v_r}\pi\,$};
		\node [style=Z phase dot] (44) at (-9, 2) {\footnotesize $\,s_{v_\ell}\pi\,$};
		\node [style=none] (45) at (-7.75, 3.25) {\footnotesize $v_\ell$};
		\node [style=none] (46) at (-7.5, -3.25) {\footnotesize $v_r$};
		\node [style=none] (47) at (-9, 0) {\footnotesize $\vdots$};
		\node [style=none] (48) at (-11.125, 3.5) {\footnotesize $f(v)$};
		\node [style=none] (49) at (-10.75, 2) {\footnotesize $u_1$};
		\node [style=none] (50) at (-10.75, -3.5) {\footnotesize $u_k$};
		\node [style=none] (51) at (5, 1.875) {\footnotesize $\ldots$};
		\node [style=none] (51a) at (5, 0.4125) {\footnotesize $\ldots$};
		\node [style=none] (52a) at (5, -0.25) {\footnotesize $\ldots$};
		\node [style=none] (52) at (5, -1.875) {\footnotesize $\ldots$};
	\end{pgfonlayer}
	\begin{pgfonlayer}{edgelayer}
		\draw (8) to (0);
		\draw [bend left] (8) to (9);
		\draw [bend left] (9) to (11.center);
		\draw [bend right] (8) to (10);
		\draw [bend right] (10) to (14.center);
		\draw [bend left] (10) to (13.center);
		\draw [bend right] (9) to (12.center);
		\draw (0) to (15);
		\draw [in=0, out=135] (15) to (16);
		\draw [in=0, out=-135] (15) to (17);
		\draw [in=0, out=155] (17) to (19.center);
		\draw [in=0, out=-135] (17) to (18.center);
		\draw [in=0, out=-155] (16) to (20.center);
		\draw [in=0, out=135] (16) to (21.center);
		\draw [in=135, out=0] (37.center) to (35);
		\draw [in=0, out=-135] (35) to (38.center);
		\draw (35) to (42.center);
		\draw [in=198, out=18] (36) to (41.center);
		\draw [in=0, out=135] (36) to (39.center);
		\draw [in=0, out=-135] (36) to (40.center);
	\end{pgfonlayer}
\end{tikzpicture}

%% file: figures/Proj.tikz
\begin{tikzpicture}
	\begin{pgfonlayer}{nodelayer}
		\node [style=none] (3) at (-1.25, 0) {};
		\node [style=X phase dot] (4) at (0, 0) {$\alpha$};
		\node [style=none] (5) at (0.625, 0.5) {\footnotesize $v$};
	\end{pgfonlayer}
	\begin{pgfonlayer}{edgelayer}
		\draw (3.center) to (4);
	\end{pgfonlayer}
\end{tikzpicture}

%% file: figures/H_rot_meas.tikz
\begin{tikzpicture}
	\begin{pgfonlayer}{nodelayer}
		\node [style=none] (3) at (-3, 0) {};
		\node [style=X phase dot] (4) at (0, 0) {$s_{v'}\pi$};
		\node [style=X phase dot] (1) at ($(4) + (-1.75,0)$) {$\alpha$};
		\node [style=none] (5) at ($(4.east) + (0.625, 0.5)$) {\footnotesize $v\smash'$};
		\node [style=none] (6) at ($(1.east) + (0.625, 0.5)$) {\footnotesize $v$};
	\end{pgfonlayer}
	\begin{pgfonlayer}{edgelayer}
		\draw (3.center) to (4);
	\end{pgfonlayer}
\end{tikzpicture}

%% file: figures/Mcorr-C.tikz
\begin{tikzpicture}
	\begin{pgfonlayer}{nodelayer}
		\node [style=none] (3) at (-2.25, 0) {};
		\node [style=X phase dot] (4) at (0, 0) {$\;\Theta(\alpha,S,T)\Big.\;$};
		\node [style=none] (5) at ($(4.east) + (0.625,0.5)$) {\footnotesize $t$};
		\node [style=none] (6) at (2.25, 0) {};
	\end{pgfonlayer}
	\begin{pgfonlayer}{edgelayer}
		\draw (3.center) to (4.center);
		\draw (4.center) to (6.center);
	\end{pgfonlayer}
\end{tikzpicture}

%% file: figures/Mcorr-C-2.tikz
\begin{tikzpicture}
	\begin{pgfonlayer}{nodelayer}
		\node [style=none] (3) at (-3, 0) {};
		\node [style=X phase dot] (4) at (0, 0) {$\;\Theta(\alpha,S \!\mathbin\Delta\! P_{u,v}\;\!,T)\Big.\;$};
		\node [style=none] (5) at ($(4.east) + (0.625,0.5)$) {\footnotesize $t$};
		\node [style=none] (6) at (3, 0) {};
	\end{pgfonlayer}
	\begin{pgfonlayer}{edgelayer}
		\draw (3.center) to (4.center);
		\draw (4.center) to (6.center);
	\end{pgfonlayer}
\end{tikzpicture}

%% file: figures/Mcorr-OddC.tikz
\begin{tikzpicture}
	\begin{pgfonlayer}{nodelayer}
		\node [style=none] (3) at (-2.25, 0) {};
		\node [style=Z phase dot] (4) at (0, 0) {$\;\Theta(\beta,S',T')\Big.\;$};
		\node [style=none] (5) at ($(4.east) + (0.625, 0.5)$) {\footnotesize $w$};
		\node [style=none] (6) at (2.25, 0) {};
	\end{pgfonlayer}
	\begin{pgfonlayer}{edgelayer}
		\draw (3.center) to (4.center);
		\draw (4.center) to (6.center);
	\end{pgfonlayer}
\end{tikzpicture}

%% file: figures/Mcorr-OddC-2.tikz
\begin{tikzpicture}
	\begin{pgfonlayer}{nodelayer}
		\node [style=none] (3) at (-3, 0) {};
		\node [style=Z phase dot] (4) at (0, 0) {$\;\Theta(\beta,S',T'\!\mathbin\Delta\!P_{u,v})\Big.\;$};
		\node [style=none] (5) at ($(4.east) + (0.625, 0.5)$) {\footnotesize $w$};
		\node [style=none] (6) at (3, 0) {};
	\end{pgfonlayer}
	\begin{pgfonlayer}{edgelayer}
		\draw (3.center) to (4.center);
		\draw (4.center) to (6.center);
	\end{pgfonlayer}
\end{tikzpicture}

%% file: figures/Pcorr-C-2.tikz
\begin{tikzpicture}
	\begin{pgfonlayer}{nodelayer}
		\node [style=none] (3) at (-3, 0) {};
		\node [style=X phase dot] (4) at (0, 0) {$\;\Theta(\alpha,S \!\mathbin\Delta\! \{v\}\;\!,T)\Big.\;$};
		\node [style=none] (5) at ($(4.east) + (0.625,0.5)$) {\footnotesize $t$};
		\node [style=none] (6) at (3, 0) {};
	\end{pgfonlayer}
	\begin{pgfonlayer}{edgelayer}
		\draw (3.center) to (4.center);
		\draw (4.center) to (6.center);
	\end{pgfonlayer}
\end{tikzpicture}

%% file: figures/Pcorr-OddC-2.tikz
\begin{tikzpicture}
	\begin{pgfonlayer}{nodelayer}
		\node [style=none] (3) at (-3, 0) {};
		\node [style=Z phase dot] (4) at (0, 0) {$\;\Theta(\beta,S',T'\!\mathbin\Delta\!\{v\})\Big.\;$};
		\node [style=none] (5) at ($(4.east) + (0.625, 0.5)$) {\footnotesize $w$};
		\node [style=none] (6) at (3, 0) {};
	\end{pgfonlayer}
	\begin{pgfonlayer}{edgelayer}
		\draw (3.center) to (4.center);
		\draw (4.center) to (6.center);
	\end{pgfonlayer}
\end{tikzpicture}

%% file: figures/ZX-fusion02.tikz
\begin{tikzpicture}[scale=0.5]
	\begin{pgfonlayer}{nodelayer}
		\node [style=Z dot, minimum width=7.5pt, minimum height=7.5pt] (0) at (0, 0) {};
		\node [style=none] (2) at (2, 2) {};
		\node [style=none] (3) at (2, -2) {};
	\end{pgfonlayer}
	\begin{pgfonlayer}{edgelayer}
		\draw [in=180, out=165, looseness=2.75] (0) to (2.center);
		\draw [in=180, out=-165, looseness=2.75] (0) to (3.center);
	\end{pgfonlayer}
\end{tikzpicture}

%% file: figures/ZX-fusion01.tikz
\begin{tikzpicture}[scale=0.5]
	\begin{pgfonlayer}{nodelayer}
		\node [style=Z dot, minimum width=7.5pt, minimum height=7.5pt] (0) at (0, 0) {};
		\node [style=Z dot, minimum width=7.5pt, minimum height=7.5pt] (1) at (1.5, 0) {};
		\node [style=none] (2) at (2, 2) {};
		\node [style=none] (3) at (2, -2) {};
	\end{pgfonlayer}
	\begin{pgfonlayer}{edgelayer}
		\draw (1) to (0);
		\draw [in=180, out=165, looseness=2.75] (0) to (2.center);
		\draw [in=180, out=-165, looseness=2.75] (0) to (3.center);
	\end{pgfonlayer}
\end{tikzpicture}

%% file: figures/ZX-fusion12.tikz
\begin{tikzpicture}[scale=0.5]
	\begin{pgfonlayer}{nodelayer}
		\node [style=Z dot, minimum width=7.5pt, minimum height=7.5pt] (0) at (0, 0) {};
		\node [style=none] (2) at (2, 2) {};
		\node [style=none] (3) at (2, -2) {};
		\node [style=X dot, inner sep=0pt, minimum width=7.5pt, minimum height=7.5pt] (4) at (-1, -1.625) {$\scriptstyle\pi$};
	\end{pgfonlayer}
	\begin{pgfonlayer}{edgelayer}
		\draw [in=180, out=165, looseness=2.75] (0) to (2.center);
		\draw [in=180, out=-165, looseness=2.75] (0) to (3.center);
	\end{pgfonlayer}
\end{tikzpicture}

%% file: figures/ZX-fusion11.tikz
\begin{tikzpicture}[scale=0.5]
	\begin{pgfonlayer}{nodelayer}
		\node [style=Z dot, minimum width=7.5pt, minimum height=7.5pt] (0) at (0, 0) {};
		\node [style=Z dot, minimum width=7.5pt, minimum height=7.5pt] (1) at (1.5, 0) {};
		\node [style=none] (2) at (2, 2) {};
		\node [style=none] (3) at (2, -2) {};
		\node [style=X dot, inner sep=0pt, minimum width=7.5pt, minimum height=7.5pt] (4) at (-1, -1.625) {$\scriptstyle\pi$};
	\end{pgfonlayer}
	\begin{pgfonlayer}{edgelayer}
		\draw (1) to (0);
		\draw [in=180, out=165, looseness=2.75] (0) to (2.center);
		\draw [in=180, out=-165, looseness=2.75] (0) to (3.center);
	\end{pgfonlayer}
\end{tikzpicture}

%% file: figures/ZX-fusionPF.tikz
\begin{tikzpicture}[scale=0.5]
	\begin{pgfonlayer}{nodelayer}
		\node [style=Z dot, minimum width=7.5pt, minimum height=7.5pt] (0) at (0, 0) {};
		\node [style=none] (2) at (2, 2) {};
		\node [style=none] (3) at (2, -2) {};
		\node [style=X phase dot] (4) at (-1, -1.625) {$s_u\pi$};
		\node [style=none] (5) at (1, 0.75) {$u$};
		\node [style=Z dot, minimum width=7.5pt, minimum height=7.5pt] (6) at (2, 0) {};
	\end{pgfonlayer}
	\begin{pgfonlayer}{edgelayer}
		\draw [in=180, out=165, looseness=2.75] (0) to (2.center);
		\draw [in=180, out=-165, looseness=2.75] (0) to (3.center);
		\draw (0) to (6);
	\end{pgfonlayer}
\end{tikzpicture}

%% file: figures/Mcorr.tikz
\begin{tikzpicture}
	\begin{pgfonlayer}{nodelayer}
		\node [style=none] (3) at (-2.25, 0) {};
		\node [style=Z phase dot] (4) at (0, 0) {$\;\Theta(\alpha,S,T)\Big.\;$};
		\node [style=none] (5) at ($(4.east) + (0.625, 0.5)$) {\footnotesize $t$};
		\node [style=none] (6) at (2.25, 0) {};
	\end{pgfonlayer}
	\begin{pgfonlayer}{edgelayer}
		\draw (3.center) to (4.center);
		\draw (4.center) to (6.center);
	\end{pgfonlayer}
\end{tikzpicture}

%% file: figures/corr_s2R.tikz
\begin{tikzpicture}
	\begin{pgfonlayer}{nodelayer}
		\node [style=none] (3) at (-4.5, 0) {};
		\node [style=Z phase dot] (4) at (0, 0) {$\;\Theta(\alpha,S \!\setminus\! \{ v \}, T)\Big.\;$};
		\node [style=none] (5) at ($(4.east) + (0.625, 0.5)$) {\footnotesize $t$};
		\node [style=none] (6) at (4.5, 0) {};
		\node [style=X phase dot] (7) at ($(4) + (3.5, 0)$) {$s_v\pi$};
		\node [style=X phase dot] (8) at ($(4) + (-3.5,0)$) {$s_v\pi$};
	\end{pgfonlayer}
	\begin{pgfonlayer}{edgelayer}
		\draw (3.center) to (4.center);
		\draw (4.center) to (6.center);
	\end{pgfonlayer}
\end{tikzpicture}

%% file: figures/Mcorr_pion2.tikz
\begin{tikzpicture}
	\begin{pgfonlayer}{nodelayer}
		\node [style=none] (3) at (-2.25, 0) {};
		\node [style=Z phase dot] (4) at (0, 0) {$\;\Theta(\alpha,\emptyset,T)\Big.\;$};
		\node [style=none] (5) at ($(4.east) + (0.625, 0.5)$) {\footnotesize $t$};
		\node [style=none] (6) at (2.25, 0) {};
	\end{pgfonlayer}
	\begin{pgfonlayer}{edgelayer}
		\draw (3.center) to (4.center);
		\draw (4.center) to (6.center);
	\end{pgfonlayer}
\end{tikzpicture}

%% file: figures/corr_s2R_pi.tikz
\begin{tikzpicture}
	\begin{pgfonlayer}{nodelayer}
		\node [style=none] (3) at (-4, 0) {};
		\node [style=Z phase dot] (4) at (0, 0) {$\;\Theta(\alpha,\emptyset, T)\Big.\;$};
		\node [style=none] (5) at ($(4.east) + (0.625, 0.5)$) {\footnotesize $t$};
		\node [style=none] (6) at (4, 0) {};
		\node [style=X phase dot] (7) at ($(4) + (3, 0)$) {$s_v\pi$};
		\node [style=X phase dot] (8) at ($(4) + (-3,0)$) {$s_v\pi$};
	\end{pgfonlayer}
	\begin{pgfonlayer}{edgelayer}
		\draw (3.center) to (4.center);
		\draw (4.center) to (6.center);
	\end{pgfonlayer}
\end{tikzpicture}

%% file: figures/corr_s2R_pion2.tikz
\begin{tikzpicture}
	\begin{pgfonlayer}{nodelayer}
		\node [style=none] (3) at (-4.5, 0) {};
		\node [style=Z phase dot] (4) at (0, 0) {$\;\Theta(\alpha,\emptyset, T \!\mathbin\Delta\! \{v\})\Big.\;$};
		\node [style=none] (5) at ($(4.east) + (0.625, 0.5)$) {\footnotesize $t$};
		\node [style=none] (6) at (4.5, 0) {};
		\node [style=X phase dot] (7) at ($(4) + (3.5, 0)$) {$s_v\pi$};
		\node [style=X phase dot] (8) at ($(4) + (-3.5,0)$) {$s_v\pi$};
	\end{pgfonlayer}
	\begin{pgfonlayer}{edgelayer}
		\draw (3.center) to (4.center);
		\draw (4.center) to (6.center);
	\end{pgfonlayer}
\end{tikzpicture}